\documentclass[aps,pra,twocolumn,10pt,superscriptaddress,longbibliograhy,floatfix]{revtex4-2}
\usepackage[utf8]{inputenc}
\usepackage[american]{babel}
\usepackage[normalem]{ulem}
\usepackage{mathpazo}
\usepackage{textcomp}
\usepackage{amsmath}
\usepackage{amssymb}
\usepackage{graphicx}
\usepackage{physics}
\usepackage{stackrel}
\usepackage{multirow}
\usepackage{braket,mleftright}
\usepackage{empheq}
\usepackage{subfigure}
\usepackage{soul}
\usepackage{blkarray}
\usepackage{bm}
\usepackage{bbold}
\usepackage{xcolor}
\usepackage{algpseudocode}
\usepackage[ruled,vlined,linesnumbered]{algorithm2e}
\usepackage[unicode=true,
 bookmarks=false,
 breaklinks=false,pdfborder={0 0 1},backref=false,colorlinks=false]
 {hyperref} 
\usepackage{dsfont}
\usepackage{upgreek}
\usepackage{enumitem}
\usepackage[most]{tcolorbox}
\usepackage{cancel}
\usepackage[export]{adjustbox}

\makeatletter

\usepackage{amsfonts}
\usepackage{amsthm}
\usepackage{graphics}
\usepackage{times}
\usepackage{color}
\usepackage{latexsym}
\usepackage{palatino}
\usepackage{bbm}
\usepackage{xparse}
\usepackage[framemethod=TikZ]{mdframed}

\usepackage{indentfirst}
\usepackage{tabularx}
\usepackage{nicematrix,tikz}
\usepackage{diagbox}
\usetikzlibrary{arrows.meta}

\newenvironment{subroutine}[1][htb]
  {
   \begin{algorithm}[#1]%
  }{\end{algorithm}}

\newcommand{\quotes}[1]{``#1''}
\newcommand{\eq}[1]{(\ref{eq:#1})}
\newcommand{\ionq}{\textrm{IonQ }}
\newcommand{\nisq}{\textrm{NISQ }}
\newcommand{\CNOT}{\operatorname{CNOT}}
\newcommand{\CNOTs}{\operatorname{CNOTs}}
\newcommand{\RBS}{\operatorname{RBS}}
\newcommand{\gRBS}{\operatorname{gRBS}}
\newcommand{\LOAD}{\operatorname{Load}}
\newcommand{\vecx}{\mathbf{x}}
\newcommand{\normx}{\norm{\vecx}}
\newcommand{\vecy}{\mathbf{y}}

\newcommand{\bin}{\mathcolor{blue}{in}}
\newcommand{\binbf}{\mathcolor{blue}{\mathbf{in}}}
\newcommand{\tout}{\mathcolor{teal}{out}}
\newcommand{\toutbf}{\mathcolor{teal}{\mathbf{out}}}
\newcommand{\vctrl}{\mathcolor{violet}{ctrl}}
\newcommand{\vctrlbf}{\mathcolor{violet}{\mathbf{ctrl}}}
\newcommand{\bctrlbf}{\mathcolor{black}{\overline{\mathbf{ctrl}}}}
\DeclareMathOperator{\atantwo}{atan2}

\newtheorem{theorem}{Theorem}[section]
\newtheorem{lemma}[theorem]{Lemma}
\newtheorem{definition}[theorem]{Definition}

\makeatother

\newcommand{\TII}{\affiliation{Quantum Research Center, Technology Innovation Institute, Abu Dhabi, UAE}}
\newcommand{\UFRJ}{\affiliation{Instituto de F\'{i}sica, Universidade Federal do Rio de Janeiro, P.O. Box 68528, Rio de Janeiro, Rio de Janeiro 21941-972, Brazil}}
\newcommand{\UB}{\affiliation{Departament de F\'{i}sica Qu\`{a}ntica i Astrof\'{i}sica and Institut de Ci\`{e}ncies del Cosmos (ICCUB), Universitat de Barcelona, Barcelona, Spain.}}

\begin{document}

\title{Quantum encoder for fixed Hamming-weight subspaces}

\author{Renato M. S. Farias}
\email{renato.msf@gmail.com}
\TII
\UFRJ

\author{Thiago O. Maciel}
\TII

\author{Giancarlo Camilo}
\TII

\author{Ruge Lin}
\TII
\UB

\author{Sergi Ramos-Calderer}
\TII
\UB

\author{Leandro Aolita}
\TII

\begin{abstract}
We present an exact $n$-qubit computational-basis amplitude encoder of real- or complex-valued data vectors of $d=\binom{n}{k}$ components into a subspace of fixed Hamming weight $k$. 
This represents a polynomial space compression of degree $k$. 
The circuit is optimal in that it expresses an arbitrary data vector using only $d-1$ (controlled) Reconfigurable Beam Splitter (RBS) gates and is constructed by an efficient classical algorithm that sequentially generates all bitstrings of weight $k$ and identifies the gates that superpose the corresponding states with the correct amplitudes. 
An explicit compilation into CNOTs and single-qubit gates is presented, with the total CNOT-gate count of $\mathcal{O}(k\, d)$ provided in analytical form. 
In addition, we show how to load data in the binary basis by sequentially stacking encoders of different Hamming weights using $\mathcal{O}(d\,\log(d))$ CNOT gates. 
Moreover, using generalized RBS gates that mix states of different Hamming weights, we extend the construction to efficiently encode arbitrary sparse vectors.
Experimentally, we perform a proof-of-principle demonstration of our scheme on a commercial trapped-ion quantum computer. 
We successfully upload a $q$-Gaussian probability distribution in the non-log-concave regime with $n = 6$ and $k = 2$.
We also showcase how the effect of hardware noise can be alleviated by quantum error mitigation.
Numerically, we show how our encoder can improve the performance of variational quantum algorithms for problems that include particle-preserving symmetries.
Our results constitute a versatile framework for quantum data compression with various potential applications in fields such as quantum chemistry, quantum machine learning, and constrained combinatorial optimizations.
\end{abstract}

\maketitle

\section{Introduction}
\label{sec:introduction}

Amplitude encoding schemes in the basis of Hamming-weight-$k$ (HW-$k$) states of $n$ qubits correspond to an interesting regime of polynomial space compression $n\in\mathcal{O}\big(k\,d^{1/k}\big)$ for data vectors of size $d$ \cite{Anselmetti2021, Arrazola2022, Monbroussou2023, Cherrat2024}. 
They are the middle-ground between two distinct encoding scenarios. 
On one end, there is an amplitude encoding scheme of zero compression that efficiently loads data in the unary basis \cite{RamosCalderer2021, Johri2021, Kerenidis2022, Kerenidis2022subspace, Zhang2023}. 
On the other end, there are techniques to load amplitudes in binary basis, which provides exponential compression \cite{Ventura1999, Long2001, Grover2002, Lloyd2018, Dallaire2018, Mitarai2019, Holmes2020, Araujo2021, Jaques2023}, but requires a number of two-qubit gates that is exponential in $n$ \cite{Barenco1995, Kitaev1997, Boykin1999, Plesch2011}.
Moreover, HW-$k$ encoders have the additional benefit of being the natural subspace for applications in fields such as quantum chemistry, due to the particle-preserving nature of typical Hamiltonians \cite{Anselmetti2021, Huggins2021, Arrazola2022, Gibbs2022, Zhao2023}.
They have also been used in the context of quantum machine learning \cite{Johri2021, Kerenidis2022, Kerenidis2022subspace, Monbroussou2023, Jain2024, Cherrat2024} and quantum finance \cite{RamosCalderer2021, Zhang2023, He2023, Cherrat2023}. 
When used as variational circuits, they take advantage of the reduced subspace to mitigate the barren plateau problem \cite{Ragone2024, Cherrat2023, Jain2024}.

Current proposals of HW-$k$ encoders present some drawbacks in practice.
For instance, Refs. \cite{Kerenidis2022subspace, Monbroussou2023} use the aforementioned unary-basis encoders as subroutines. 
While allowing for the generalization to HW $k > 1$ using only two-qubit Givens rotations, these encoders require orthogonalization of the input data \cite{Kerenidis2022} or solving non-linear systems of $\binom{n}{k} - 1$ equations to compute the rotation angles \cite{Monbroussou2023}. 
Here we address these classical overheads while still performing an exact encoding.

We present an exact, ancilla-free amplitude encoder for real and complex data in HW-$k$ subspaces. 
We detail a classical algorithm that, given an initial bitstring of Hamming weight $k$ and a data vector of size $d$, outputs instructions for a sequence of (controlled) Givens rotations and their respective angles in hyperspherical coordinates of the data vector, hence efficient to compute. 
An explicit compilation into $\CNOTs$ and single-qubit gates is presented, with the total $\CNOT$-gate count of $\mathcal{O}(k\, d)$ provided in analytical form. 
Moreover, our HW-$k$ encoder can be used as a subroutine to power more complex algorithms. 
For instance, by combining all Hamming weights $k \leq n$ we obtain a binary encoder using $\mathcal{O}(d\,\log(d))$ $\CNOT$ gates. 
Using a generalized RBS gate designed to superpose states with different Hamming weights, we extend our methods and derive a procedure to load data in the paradigm of the classical sparse-access model.
Experimentally, we deployed a proof-of-concept instance on the \texttt{Aria-1} ion-trap quantum processor from \ionq \cite{IONQ}, uploading a $q$-Gaussian probability distribution. 
The experimental results are further enhanced using a well-known error mitigation technique called Clifford Data Regression \cite{czarnik2021, Lowe2021}.
Numerically, we performed two experiments: first, we analyzed the effects of circuit noise in the fidelity of the encoder under a simple noise model;
then, we showed how the state encoder can also be used as a variational quantum ansatz, being a promising candidate of ansatz for problems with particle-preserving symmetry.

The paper is structured as follows. 
In Sec. \ref{sec:preliminaries} we introduce the gates used in our algorithm. 
Then, we present our framework for HW-$k$ encoders in Sec. \ref{sec:hamming_weight}. 
Secs. \ref{sec:sparse} and  \ref{sec:binary}, respectively, expand on the sparse and binary encoders using the fixed HW-$k$ encoder as a subroutine. 
In Sec. \ref{sec:hardware}, we show the deployment of our HW-$k$ encoder on quantum hardware, and we conclude in Sec. \ref{sec:conclusions}.

\section{Preliminaries}
\label{sec:preliminaries}

\paragraph{Notation.} Let $b\in\{0,1\}^{\otimes n}$ be a bitstring of length $n$ and $\abs{b}$ its Hamming weight (HW); given two bitstrings $b$ and $b^{\prime}$, their Hamming distance is $\abs{b \oplus b^{\prime}}$, where $\oplus$ is addition mod $2$. For convenience, we refer to the HW of a computational basis state $\ket{b}$ as being the HW of the bitstring $b$, and both are used interchangeably. 
The value of the $j$-th bit of a bitstring $b$ is denoted by $b(j)$. 
Given a set $\mathcal{I}$, we call $\mathcal{I}_c$ its complementary set. We also use the shorthand notation $[d] \coloneqq \{1, 2, \ldots, d\}$. 

\begin{figure}[t!]
\centering
\includegraphics[scale=0.67]{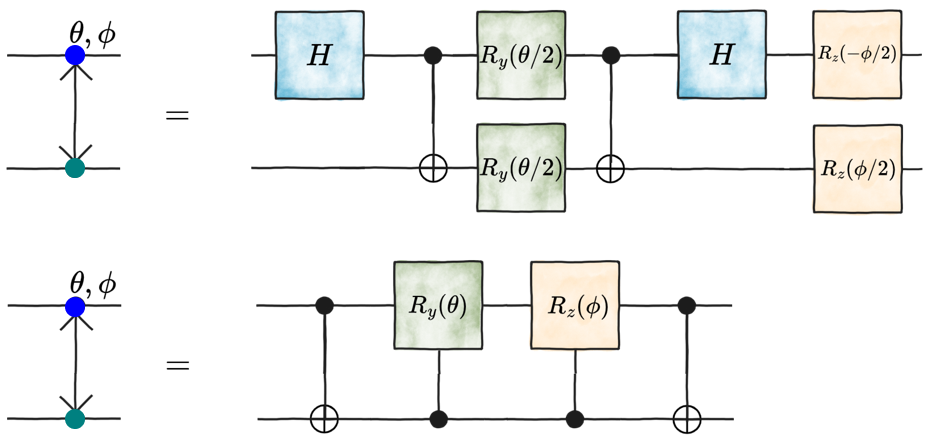}
\caption{
    \textbf{Two possible compilations of the complex $\RBS$ gate $R^{\bin}_{\tout}(\theta, \, \phi)$.} 
    Colors indicate the input (blue) and output (green) qubits.
    \textit{(Top}) Compilation using two $\CNOTs$ and single-qubit gates; 
    \textit{(Bottom)} Compilation using two $\CNOTs$ and two controlled-Pauli rotations; 
    these rotations can be merged into a single controlled-$\operatorname{SU}(2)$ gate (see App. \ref{app:compilation_one_gate}). 
} 
\label{fig:rbscomplex}
\end{figure}

Here we introduce gates that superpose two computational basis states with tunable amplitudes. These will be the building blocks for our encoders. 
We define the single-qubit Pauli-$Y$ rotation as $R_{y}(\theta) \coloneqq e^{-i\theta Y}$.
This gate superposes the two computational basis states $\ket{0}$ and $\ket{1}$ of a single-qubit subspace by acting as 
\begin{equation}\label{eq:RY}
\begin{aligned}
    R_{y}(\theta)\ket{0} & = \cos(\theta) \ket{0} + \sin(\theta) \ket{1} \\
    R_{y}(\theta)\ket{1} & = \cos(\theta) \ket{1} - \sin(\theta) \ket{0}\,.
\end{aligned}
\end{equation}
When acting on a multi-qubit state, (controlled) $R_{y}(\theta)$ rotations can be used to superpose any two computational basis states of Hamming distance $1$.

Another gate of interest is one creating a superposition between two computational basis states of Hamming distance $2$ by acting only on a two-qubit subspace. 
Denoting by $\bin$ ($\tout$) the bit that has value $1$ ($0$) in the first state and $0$ ($1$) in the second, we define the two-qubit gate $R^{\bin}_{\tout}(\theta, \, \phi)$ acting on qubits $\bin$ and $\tout$ as
\begin{equation}\label{eq:complexRBS}
\begin{aligned}
    R^{\bin}_{\tout}(\theta, \, \phi) \ket{10} &= e^{i \phi} \cos(\theta) \ket{10} + e^{-i \phi} \sin(\theta) \ket{01} \\
    R^{\bin}_{\tout}(\theta, \, \phi) \ket{01} &= e^{-i \phi} \cos(\theta) \ket{01} - e^{i \phi} \sin(\theta) \ket{10}
\end{aligned}    
\end{equation}
and as the identity elsewhere. 
We call $R^{\bin}_{\tout}(\theta, \, \phi)$ a \emph{complex Reconfigurable Beam Splitter} ($\RBS$) gate,  reducing to the usual $\RBS$ gate for $\phi=0$ \footnote{The gate $R^{\bin}_{\tout}(-\theta,\,0)$ is preferred by some authors and referred to as a Givens rotation.}. 
We also denote $R^{\bin}_{\tout}(\theta, \, 0)$ as $R^{\bin}_{\tout}(\theta)$. 
In Fig. \ref{fig:rbscomplex}, we present two possible compilations of the complex $\RBS$ into $\CNOTs$, $R_{y}(\theta)$, and $R_z(\phi) \coloneqq e^{-i\phi Z}$ gates.
In App. \ref{app:compilation_one_gate}, we show how to reduce the number of controlled operations in the second compilation by merging the two controlled-Pauli rotations into one controlled-$\operatorname{SU}(2)$ rotation.

Def. \ref{lemma:gRBS} below generalizes Eq. \eq{complexRBS} for two arbitrary computational basis states of possibly different HWs.

\begin{definition}[Generalized complex RBS gate ($\gRBS$)]\label{lemma:gRBS}
Let $m,m^{\prime}\ge1$ be integers, $\binbf:=\{\mathcolor{blue}{in_1,\ldots,in_m}\}\subset[m+m^{\prime}]$, and $\toutbf:=\{\mathcolor{teal}{out_1,\ldots,out_{m^{\prime}}}\}:=[m+m^{\prime}]\setminus\binbf$. Let $\ket{1_{\binbf}\,0_{\toutbf}}$ ($\ket{0_{\binbf}\,1_{\toutbf}}$) stand for the computational basis state having $m$ ones (zeros) at bit positions $\binbf$ and $m^{\prime}$ zeros (ones) at bit positions $\toutbf$.  
We define the generalized complex RBS gate $R^{\binbf}_{\toutbf}(\theta, \, \phi)$ relative to the input and output sets $\binbf$ and $\toutbf$ as the $(m+m')$-qubit gate acting as
\begin{gather}\label{eq:gR}
\begin{aligned}
    \!R^{\binbf}_{\toutbf}&\!(\theta, \, \phi)\!\ket{1_{\binbf}\,0_{\toutbf}} \!= \\ &e^{i \phi}\cos(\theta)\!\ket{1_{\binbf}\,0_{\toutbf}} + e^{-i \phi}\sin(\theta)\!\ket{0_{\binbf}\,1_{\toutbf}}\! \\
    \!R^{\binbf}_{\toutbf}&\!(\theta, \, \phi)\!\ket{0_{\binbf}\,1_{\toutbf}} \!= \\ &e^{-i \phi}\cos(\theta) \!\ket{0_{\binbf}\,1_{\toutbf}} - e^{i \phi}\sin(\theta) \!\ket{1_{\binbf}\,0_{\toutbf}} \!
\end{aligned}        
\end{gather}
on the 2-dimensional subspace spanned by $\{\ket{1_{\binbf}\,0_{\toutbf}},\ket{0_{\binbf}\,1_{\toutbf}}\}$, and as the identity elsewhere. 
\end{definition}

\begin{figure}[t!]
\centering
\includegraphics[scale=0.61]{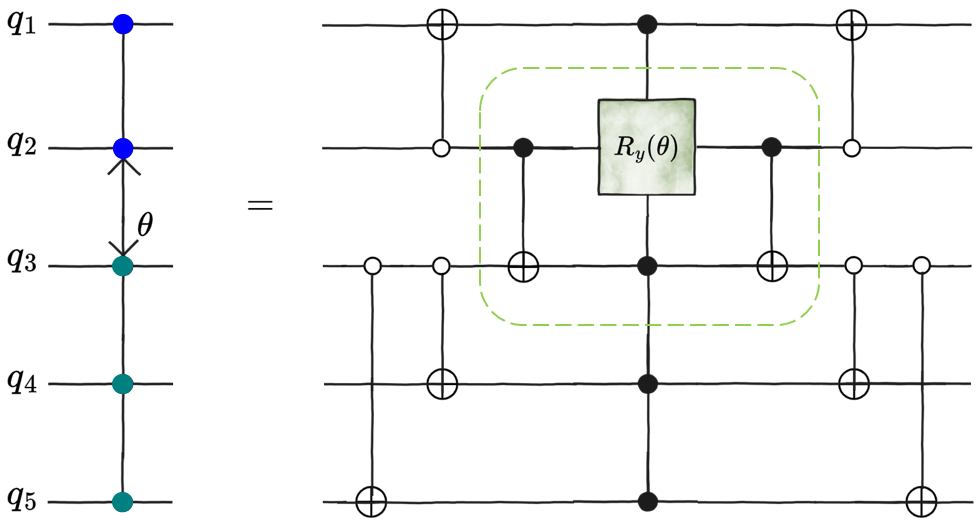}
\caption{
    \textbf{Possible compilation of the gRBS gate.} 
    An illustration of the $(m + m')$-qubit $\gRBS$ gate for $m = 2$ and $m^{\prime} = 3$, namely $R^{\mathcolor{blue}{1,2}}_{\mathcolor{teal}{3,4,5}}(\theta,\phi=0)$. 
    This gate maps $\ket{11000}$ onto a superposition of $\ket{11000}$ and $\ket{00111}$ with real amplitudes, according to Eq. \eq{gR}.
    Colors indicate input (blue) and output (green) qubits.
    We chose the compilation in Fig. \ref{fig:rbscomplex} (\textit{Bottom}) (green dashed box). 
    Controlled-$\gRBS$ gates can be built by adding extra controls to the controlled-$R_{y}$ rotation.
    This compilation can be generalized to $R^{\mathcolor{blue}{1,2}}_{\mathcolor{teal}{3,4,5}}(\theta, \, \phi)$ by replacing the $R_{y}$ gate with a generic $\operatorname{SU}(2)$ gate
    (see App. \ref{app:compilation_one_gate}).
} 
\label{fig:generalizedRBS}
\end{figure}

\noindent Whenever needed, we will write down the gate addresses explicitly, \emph{i.e.}, $R^{\binbf}_{\toutbf}(\theta, \, \phi) \equiv R^{\mathcolor{blue}{in_1,\ldots,in_m}}_{\mathcolor{teal}{out_1, \ldots, out_{m^{\prime}}}}(\theta, \, \phi)$. 
For $m'=m$, the gRBS gate in Eq. \eq{gR} is HW-preserving, reducing to Eq. \eq{complexRBS} if $m'=m=1$. 
The latter will be the basic building block of our fixed-Hamming-weight encoder for dense data in Sec. \ref{sec:hamming_weight}. 
$\gRBS$ gates with generic $m,m^\prime$ superpose computational basis states of possibly different HWs and will be used to deal with sparse data in Sec. \ref{sec:sparse} and to build a binary basis encoder in Sec. \ref{sec:binary}. 

In Fig. \ref{fig:generalizedRBS}, we present a possible compilation of the $\gRBS$ in Eq. \eq{gR} when $\phi=0$ (see App. \ref{app:compilation_one_gate} for the generalization to $\phi\ne0$). 
The compilation goes as follows: 
$(i)$ choose any two qubits, one from the $\bin$ and one from the $\tout$ sets, in which to act with a $\RBS$ gate -- in this example, $R_{\mathcolor{teal}{3}}^{\mathcolor{blue}{2}}(\theta)$ (green dashed box);
$(ii)$ add a network of $m+m^\prime-2$ anti-$\CNOT$ gates with the anti-controls on each qubit selected in $(i)$ and the targets on all the other qubits from the respective set ($\bin$ or $\tout$); 
$(iii)$ add the $\RBS$ gate chosen in $(i)$ controlled by all the remaining qubits from $\bin,\tout$ -- in this example, $c_{\mathcolor{blue}{1},\mathcolor{teal}{4,5}}\,R_{\mathcolor{teal}{3}}^{\mathcolor{blue}{2}}(\theta)$; 
$(iv)$ undo the anti-CNOT network of step $(ii)$. 

As explained in the following sections, our encoders will repeatedly add a new component $\ket{b^{\prime}}$ to a superposition state $\ket{\psi}$ by acting non-trivially only on the two-dimensional subspace spanned by $\ket{b^{\prime}}$ and one of the existing components $\ket{b}$ of $\ket{\psi}$. 
This requires, in general, controlled-$\gRBS$ gates, as detailed in the lemma below.

\begin{lemma}[Adding a new amplitude to a state $\ket{\psi}$]
Let $n \ge 2$ be an integer, $b,\, b^{\prime} \in\{0,1\}^{\otimes n}$ bitstrings of Hamming weights $\abs{b}$ and $\abs{b^{\prime}}$, respectively, and $\ket{\psi}$ be an arbitrary $n$-qubit state having a non-zero amplitude along the computational basis state $\ket{b}$ but not along $\ket{b'}$. 
Let $\mathcal{I},\mathcal{I}'$ be the sets of bits where $b$ and $b^{\prime}$ have a $1$, respectively, and define the sets of input, output, control, and anti-control qubits respectively by $\binbf:=\mathcal{I} \setminus \mathcal{I}^{\prime}$, $\toutbf \coloneqq \mathcal{I}^{\prime} \setminus \mathcal{I}$, $\vctrlbf \coloneqq \mathcal{I} \cap \mathcal{I}^{\prime}$, and $\bctrlbf \coloneqq \mathcal{I}_{c} \cap \mathcal{I}_{c}^{\prime}$. 
Their respective cardinalities are denoted $m$, $m^{\prime}$, $\ell$, and $\ell^{\prime}$, with $m+m^{\prime}+\ell+\ell^{\prime}=n$.  
Then, the $n$-qubit multi-controlled gate $\overline{c}_{\bctrlbf} \, c_{\vctrlbf} \, R^{\binbf}_{\toutbf}(\theta, \, \phi)$ acts on $\ket{\psi}$ by superposing its $\ket{b}$ component with $\ket{b'}$ as in Eq. \eq{gR} while leaving all the other components untouched. 
In particular, 
\begin{enumerate}[leftmargin=*]
    \item if $m'\ge m$ and $\ket{b}$ is the computational basis state of maximum HW in $\ket{\psi}$, the anti-controls can be removed and the ($n-\ell'$)-qubit gate $c_{\vctrlbf}\,R^{\binbf}_{\toutbf}(\theta, \, \phi)$ superposes $\ket{b}$ with $\ket{b'}$ of HW $\abs{b'}=\abs{b}+(m'-m)$;
    \item if $m'\le m$ and $\ket{b}$ is the computational basis state of minimum HW in $\ket{\psi}$, the controls can be removed and the ($n-\ell$)-qubit gate $\overline{c}_{\bctrlbf}\,R^{\binbf}_{\toutbf}(\theta, \, \phi)$ superposes $\ket{b}$ with $\ket{b'}$ of HW $\abs{b'}=\abs{b}-(m-m')$.
\end{enumerate}
\end{lemma}
Our HW-$k$ encoders in Sec. \ref{sec:hamming_weight} fall in the scope of the first subcase of the above lemma and, therefore, will be based on multi-controlled $\gRBS$ gates. 
In App. \ref{app:compilation_one_gate} we show that a $c_{\mathcolor{violet}{ctrl_1,\ldots,ctrl_\ell}}R^{\mathcolor{blue}{in_1,\ldots,in_m}}_{\mathcolor{teal}{out_1, \ldots, out_{m^{\prime}}}}(\theta, \, \phi)$ gate can be compiled using at most $22(m + m^{\prime}) + 20 (\ell - 1)$ $\CNOT$ gates, with exact compilations provided for $\ell \le 3$ and $m = m^{\prime} = 1$, and upper bounds otherwise (see Lemma \ref{lemma:cgRBScompilationcomplex} and Table \ref{tab:cgRBScompilation} for a summary). 

\section{Hamming-weight-$k$ encoders}
\label{sec:hamming_weight}

We start this section by defining generic amplitude encoders and parameter-optimal encoders.
Next, we present our Hamming-weight-$k$ encoder.

\subsection{Amplitude encoders}
\label{sec:amplitude_encoders}

Let us define amplitude encoders for an arbitrary set of computational basis states. 

\begin{definition}[Amplitude encoder]
\label{def:amplitudeencoder}
Let $\vecx$ be any $d$-dimensional data vector, $\normx \coloneqq \sqrt{\sum_{j = 1}^{d} \abs{x_{j}}^{2}}$ be its $\ell_{2}$-norm, and $B \coloneqq \{\ket{b_{j}}:b_{j} \in \{0,1\}^{\otimes n}\}_{j \in [d]}$ be a set of $d$ computational basis states of $n$ qubits with $n \geq \log_2(d)$. 
An amplitude encoder in the basis $B$ is an $n$-qubit parameterized quantum circuit $\LOAD_{B}(\vecx)$, with gate parameters depending on $\vecx$, such that
\begin{align}\label{eq:amplitude_encoding_def}
\LOAD_{B}(\vecx) \, \ket{0}^{\otimes n} \coloneqq \frac{1}{\normx} \, \sum_{j=1}^{d} \, x_j \, \ket{b_j}.
\end{align}
\end{definition}

We explicitly distinguish between real-valued data vectors, $\vecx \in \mathbb{R}^d$, and complex-valued ones, $\vecx \in \mathbb{C}^d$, since the basic quantum gates will be distinct in each case. 
Gate parameters in Def. \ref{def:amplitudeencoder} refer to rotation angles such as $\theta$ and $\phi$ introduced in Sec. \ref{sec:preliminaries}. 
The total number of parameterized gates in a given circuit may vary depending on the specific circuit synthesis. 
In particular, in the case of variational encoders, usual schemes \cite{Monbroussou2023, Cherrat2024} resort to overparametrization to reach the necessary expressivity to encode an arbitrary vector $\vecx$. 
We will present a deterministic protocol to construct circuits using the minimum number of parameters to express any $\vecx$ exactly. 
Our construction is based on the observation that, due to the $\ell_2$-normalization in Eq. \eq{amplitude_encoding_def}, we effectively encode the vector $\vecx / \normx$, which lies in the unit ($d-1$)-sphere for $\vecx \in \mathbb{R}^{d}$ or ($2d-1$)-sphere for $\vecx \in \mathbb{C}^{d}$. 
Therefore, expressing an arbitrary data vector requires exactly $d-1$ and $2d-1$ real parameters, respectively. This motivates the following definition of a parameter-optimal encoder. 

\begin{definition}[Parameter-optimal amplitude encoder]\label{def:parameter_optimal}
An amplitude encoder $\LOAD_{B}(\vecx)$ is parameter-optimal if it uses exactly $d-1$ gate parameters for $\vecx \in \mathbb{R}^d$ or $2d-1$ gate parameters for $\vecx \in \mathbb{C}^d$. 
\end{definition}

Efficient quantum data encoders using $\order{d}$ two-qubit gates are known in the unary basis $B_{1} = \big\{ \ket{b_{j}}: b_{j} \in \{0,1\}^{\otimes n}\text{ and }\abs{b_{j}} = 1 \big\}_{j \in [d]}$ \cite{RamosCalderer2021, Johri2021, Kerenidis2022, Kerenidis2022subspace, Zhang2023}. 
However, they require as many qubits as data entries ($n=d$), hence not offering any space compression. 
The binary basis $B_{\text{binary}} = \big\{ \ket{b_{j}}: b_{j} \in \{0,1\}^{\otimes n} \big\}$, with $n = \lceil \log_{2}(d) \rceil$, gives rise to exponential space compression, but the corresponding encoders have exponential gate complexity \cite{Plesch2011}.
In the following section, we study amplitude encoding on a Hamming-weight-$k$ subspace, \emph{i.e.}
\begin{align}\label{eq:HWk_basis}
B_{k} \coloneqq \big\{ \ket{b_{j}} : b_{j} \in \{0,1\}^{\otimes n} \text{ and } \abs{b_{j}} = k\big\}_{j \in [d]} \, ,
\end{align}
with $k \in [n]$ and $d = \binom{n}{k}$.  
The amplitude encoder $\LOAD_{B_k}(\vecx)$ corresponds to the intermediate regime of polynomial space compression, $n \in \mathcal{O}\big(k \, d^{1/k} \big)$. 
In Sec. \ref{sec:binary}, we construct a binary basis encoder combining HW-$k$ encoders for all $k$, since $B_{\text{binary}} \equiv \bigcup_{k = 0}^{n} B_{k}$. 

\begin{subroutine}[t!]
\caption{\label{alg:nxt-hw-bs} Finding the next HW-$k$ bitstring}
\DontPrintSemicolon
{\bf Input:} bitstring $b$ and indices {\it marked} of marked bits of $b$\;
{\bf Output:} new bitstring $b^{\prime}$ satisfying $\abs{b^{\prime}}=\abs{b}$ and $\abs{b\oplus b^{\prime}}=2$, and updated list of marked bits\;
\;
\SetKwFunction{NextBS}{\text{\sc{NextBS}}}
\SetKwProg{Fn}{Function}{:}
\Fn{\NextBS{b, marked}}
{\;
    $p\leftarrow$ $\max(marked)$ \Comment{index of the pivot bit}\;
    \If{$b(p)=0$}{
        $\ell \leftarrow$ index of nearest $1$ to the right of $p$\; 
    }
    \Else{
        $\ell \leftarrow$ index of farthest $0$ to the right of $p$, without passing another $1$\;
    }        
    $b^{\prime} \leftarrow b $ with $p$-th and $\ell$-th bits exchanged \;
    $j\leftarrow$ starting index of the last sequence of equal bits in $b^{\prime}$\;
    $marked' \leftarrow marked\setminus\{p\}\cup \{p+1, \ldots, j-1\}$\;\Comment{unmark bit $p$ and mark all bits between $p$ and $j$} \;
}
    \KwRet{$\{b^{\prime},\,marked^\prime\}$}
\end{subroutine}

\subsection{HW-$k$ amplitude encoder}

In this section, we introduce an amplitude encoder $\LOAD_{B_{k}}(\vecx)$ on a Hamming-weight-$k$ subspace $B_k$ given by Eq. \eq{HWk_basis}. 
Hereafter, without loss of generality, the vector $\vecx$ is assumed to have dimension $d = \binom{n}{k}$. 
The resulting quantum circuit uses as basic ingredients the HW-preserving gates $R^{\bin}_{\tout}(\theta, \, \phi)$ introduced in Sec. \ref{sec:preliminaries}. 
The circuit architecture is decided by a classical algorithm that generates a sequence $[b_{j}]_{j \in [d]}$ of all HW-$k$ bitstrings ordered such that the Hamming distance $|b_{j}\oplus b_{j+1}|$ between subsequent bitstrings equals two.  
To generate these sequences, we use a \emph{Gray code} called Ehrlich's algorithm \cite{Even1973}. 
The code receives an initial bitstring $b$ with $\abs{b} = k$ as input and outputs the sequence of all bitstrings of HW $k$ by flipping only a pair of bits at each step of an iterative procedure. The total number of iterations is exactly $\binom{n}{k}-1$.
For completeness, the algorithm for one such iterative step is presented in Alg. \ref{alg:nxt-hw-bs}.
The fact that $|b_{j}\oplus b_{j+1}|=2$ for every $j$ allows one to superpose the corresponding states $\ket{b_{j}}$ and $\ket{b_{j+1}}$ using (controlled) $\RBS$ gates, as explained in Sec. \ref{sec:preliminaries} (see Eq. \eq{complexRBS}). 
In Alg. \ref{alg:rbs-gate-params}, we present a routine that, given $b_{j}$ and $b_{j+1}$, computes the parameters $\{\bin,\, \tout,\, \vctrl \}$ of the gate needed to superpose $\ket{b_{j}}$ and $\ket{b_{j+1}}$ as follows: the $1$'s that remain unchanged between bitstrings $b_{j}$ and $b_{j+1}$ correspond to control qubits, while the bits that were $1$ ($0$) in $b_{j}$ and became $0$ ($1$) in $b_{j+1}$ are associated with the input (output) of the RBS gate. 
The angles $(\theta_j, \, \phi_j)$ are coordinates of the vector $\vecx / \normx$ on the hypersphere and depend on whether $\vecx\in\mathbb{R}^d$ or $\vecx\in\mathbb{C}^d$. 
They will be presented in Secs. \ref{sec:real_dense} and \ref{sec:complex_dense} below.

The general procedures to construct circuits for $\vecx\in\mathbb{R}^d$ and $\vecx\in\mathbb{C}^d$ are given, respectively, by Alg. \ref{alg:hw-enc} and Alg. \ref{alg:hw-enc-cplx}. 
These algorithms are classical and have four independent routines that can be summarized as follows:

\begin{tcolorbox}[colback=gray!5!white,colframe=gray!75!black,width=\columnwidth,title={{{\bf Summary of Algs. \ref{alg:hw-enc} and \ref{alg:hw-enc-cplx}}}},center title]
\begin{enumerate}[leftmargin=*]
    \item 
    Use the Ehrlich algorithm (Alg. \ref{alg:nxt-hw-bs}) to generate a sequence $[b_j]_{j\in[d]}$ of HW-$k$ bitstrings with Hamming distances $\abs{b_j\oplus b_{j+1}}=2$ for all $j$;
    \item Given $b_j$ and $b_{j+1}$, use Alg. \ref{alg:rbs-gate-params} to compute gate addresses $\vctrl$, $\bin$, $\tout$ needed to superpose states $\ket{b_j}$ and $\ket{b_{j+1}}$;
    \item If $\vecx\in\mathbb{R}^d$, compute gate parameter $\theta_j$ using Eq. \eq{spherical_coords}; 
    if $\vecx\in\mathbb{C}^d$, calculate $(\theta_j, \, \phi_j)$~ using Eqs. \eq{cplx_arg_angles_theta} and \eq{cplx_arg_angles_phi};
    \item Add gate $c_{\vctrl}R^{\bin}_{\tout}(\theta_{j}, \, \phi_{j})$ (or $c_{\vctrl}R^{\bin}_{\tout}(\theta_{j})$) to the circuit;
    \item Repeat steps $2$-$4$, $\, \forall \, j$.
\end{enumerate}
\end{tcolorbox}

\noindent{We} refer to the particular bitstring sequence generated in step $1$ of the Summary as the \emph{Ehrlich ordering} (see Tab. \ref{tab:example6choose2} in App. \ref{app:tables} for an example with $n=6$ and $k=2$). 
It is worth noting that if one is required to encode amplitudes into specific pre-established states $\ket{b_j}$ (\emph{e.g.}, preserving standard binary ordering), then one must sort the input data vector $\vecx$ beforehand to match the Ehrlich ordering. 
Hereafter, we assume that either $\vecx$ is sorted in Ehrlich ordering or that no specific ordering order is required by the data. 
When that is not the case, more complicated gates may be needed to superpose $\ket{b_j}$ and $\ket{b_{j+1}}$ --- this is the scope of the sparse-access model encoder discussed in Sec. \ref{sec:sparse}.
Also, for $k = 1$, the circuit recovers the unary encoding architecture \cite{Landman_2022}.

\begin{subroutine}[t!]
\DontPrintSemicolon
{\bf Input:} two bitstrings $b$ and $b^{\prime}$, and set \textit{untouched} of indices of bits with value $1$ in the initial bitstring\;
{\bf Output:} Gate parameters $\{\bin,\, \tout,\, \vctrl\}$ needed to superpose $\ket{b},\ket{b'}$, and updated list of untouched bits\;
\caption{\label{alg:rbs-gate-params} Getting gate addresses from $b,b^{\prime}$}
\;
\SetKwFunction{GateParams}{\text{\sc{GateAddresses}}}
\SetKwProg{Fn}{Function}{:}
\Fn{\GateParams{$b, b^{\prime}, untouched$}}
{\;
    $\mathcal{I} \leftarrow$ list of indices where $b$ has an $1$\;
    $\mathcal{I}' \leftarrow$ list of indices where $b^{\prime}$ has an $1$\;
    $\vctrl \leftarrow \mathcal{I} \cap \mathcal{I}' $ \;
    $\bin \leftarrow \mathcal{I} \setminus \vctrl $ \;
    $\tout \leftarrow \mathcal{I}' \setminus \vctrl $ \;
    $untouched^\prime \leftarrow untouched \setminus \{\bin,\tout \} $\;
    \Comment{remove $\bin,\tout$ since now they have been touched}\;
    $\vctrl \leftarrow \vctrl \setminus untouched^\prime$ \Comment{remove unnecessary controls}\;
}
    \KwRet{$\{\bin,\, \tout,\, \vctrl, untouched^\prime\}$}
\end{subroutine}

A priori, the procedure above adds a controlled $\RBS$ gate with exactly $k-1$ controls to the circuit at each iteration. 
This would result in the naive count of $\binom{n}{k}-1$ controlled $\RBS$ gates in total with $k-1$ controls each.
However, several controls can be removed at the initial iterations since the corresponding control qubits are in the state $\ket{1}$ and not having been entangled with any other qubit yet. 
The variable \emph{untouched} in Alg. \ref{alg:rbs-gate-params} keeps track of these disentangled qubits and is used to remove unnecessary controls from the list of gates returned by the naive algorithm. 
As a result of this optimization, the circuit resulting from Alg. \ref{alg:hw-enc} consists of $\binom{n - (k - \ell)}{\ell + 1}$ controlled $\RBS$ gates having $\ell$ controls each, with $\ell+1 \in [k]$ ($\ell = 0$ corresponds to removing all the controls and $\ell=k-1$ to no removal). 
This leads to a significant reduction in the total $\CNOT$-gate count of the circuit. 
The total number of parameterized gates, $\sum_{\ell = 0}^{k - 1}\binom{n - (k - \ell)}{\ell + 1} = \binom{n}{k} - 1$, remains the same, while the total $\CNOT$-gate count of the optimized HW-$k$ encoder $\LOAD_{B_k}(\vecx)$ generated by Alg. \ref{alg:hw-enc} is

\begin{gather}\label{eq:totalCNOTcount}
\begin{aligned}
\#\CNOTs &= \sum_{\ell = 0}^{k - 1} \binom{n - (k - \ell)}{\ell + 1} \, C_{\ell} \\
&\leq \left[\binom{n}{k} - 1\right] \, C_{k-1} \, ,
\end{aligned}
\end{gather}
where $C_{\ell}$ is the $\CNOT$-gate cost of a single $\ell$-controlled $R^{\bin}_{\tout}(\theta, \, \phi)$ gate.
The upper bound is the naive count without simplifying initial controls. 
The latter will be the exact CNOT count whenever the simplification is not possible because the initial state already contains a superposition, such as for the binary encoder in Sec. \ref{sec:binary}.
The precise values of $C_{\ell}$ depend on the compilation (see App. \ref{app:compilation_one_gate}) and also differ for $\phi = 0$ and $\phi \ne 0$. 
We present explicit expressions for each of these cases in Lemmas \ref{lemma:totalCNOTcount_real_dense} and \ref{lemma:totalCNOTcount_cpx_dense}. 
We observed heuristically that choosing $1^{k} 0^{n-k}$ as the initial bitstring in Alg. \ref{alg:nxt-hw-bs} maximizes the removal of redundant controls in the initial iterations. Furthermore, it has the extra benefit of grouping the controls of some of the subsequent controlled $\RBS$ gates. This opens the possibility of further reducing the total $\CNOT$ cost in the circuit using one clean ancillary qubit to control these groups of gates all at once (see App. \ref{app:cnc}). 
For simplicity, we only present $\CNOT$-gate counts for the ancilla-free implementation.

It is important to mention that, for HW $k > n/2$, one can avoid the costly implementation of ($k-1$)-controlled-RBS gates that would arise from the above procedure by exploring the fact that HW-$k$ bitstrings are the negation of bitstrings of HW-($n-k$). Therefore, a more efficient circuit can be built by copying the architecture of an HW-($n-k$) encoder circuit with the negated initial bitstring and controlled $\RBS$ gates replaced by anti-controlled $\RBS$s. However, this procedure only works if the initial state is not in a superposition. 

Finally, the way we construct the circuits leads to obtaining the angles of the RBS gates using the hyperspherical coordinates system. The next two subsections cover the case when $\vecx$ is real- and complex-valued respectively.

\begin{algorithm}[t!]
\DontPrintSemicolon
{\bf Input:} number of qubits $n$, Hamming weight $k$, and data vector $\vecx$ of dimension $d\le\binom{n}{k}$\;
{\bf Output:} Quantum circuit $\mathcal{C} = \LOAD_{B_k}(\vecx)$\;
\;
\caption{\label{alg:hw-enc} HW-$k$ encoder for dense $\vecx\in\mathbb{R}^d$}
\DontPrintSemicolon
\SetKwFunction{FMain}{\text{\sc{HWEncoder}}}
\SetKwFunction{FNextHWkBS}{\text{\sc{NextBS}}}
\SetKwFunction{GateParams}{\text{\sc{GateAddresses}}}
\SetKwProg{Fn}{Function}{:}
\Fn{\FMain{n, k, $\vecx$}}
{\;
    $\mathcal{C} \leftarrow {\rm Circuit}(n)$\;
    \For{$i \leftarrow n-k+1 \; \KwTo \; n$} {
        Add gate $X$ on qubit $i$ to $\mathcal{C}$
    } 
    $b \leftarrow 1^{k}0^{n-k}$ \Comment{initial bitstring}\; 
    $marked \leftarrow \{1, \ldots, k\}$ \Comment{indices of marked bits}\; 
    $untouched \leftarrow \{n-k+1,\ldots,n\}$ \Comment{indices of bits with value $1$ that have never been touched by an RBS gate}\;
    \For{$j \leftarrow 1 \; \KwTo \; d-1$} {
        $\{b^{\prime},marked^\prime\} \leftarrow \FNextHWkBS(b, marked)$ \;
        $\{\bin,\tout,\vctrl,untouched^\prime\} \leftarrow \GateParams(b,b^{\prime},untouched)$ \;
        Compute $\theta_j$ from Eq. \eq{spherical_coords} \;
        Add gate $c_{\vctrl}R^{\bin}_{\tout}(\theta_j)$ to $\mathcal{C}$\;
        $\{b,marked,untouched\} \leftarrow \{b^{\prime},marked^\prime,untouched^\prime\}$ \;
    }
    }
    \KwRet{$\mathcal{C}$}
\end{algorithm}

\subsubsection{Dense real-valued data}
\label{sec:real_dense}

For $\vecx\in \mathbb{R}^d$, the $d-1$ angles of the $c_{\vctrl}R^{\bin}_{\tout}(\theta_j)$ gates used to build the circuit in Alg. \ref{alg:hw-enc} are
\begin{gather}\label{eq:spherical_coords}
\begin{aligned}
\theta_{j} &\coloneqq \atantwo\left(\sqrt{\sum_{j^{\prime} = j + 1}^{d} x_{j^{\prime}}^{2}} \, ,\, x_{j} \right),\,\, j \in [d-2] \\
\theta_{d-1} &\coloneqq \atantwo \big( x_{d}\, , \, x_{d-1} \big),
\end{aligned}
\end{gather}

\noindent{where} the two-argument arctangent function $-\pi<\atantwo(y,x)\le\pi$ is the principal value of $\operatorname{arg}(x+i\,y)$. 
These are the hyperspherical coordinates of the normalized vector $\vecx / \normx \in \mathbb{S}^{d-1}$, \emph{i.e.}
\begin{gather}\label{eq:hypspherical}
\begin{aligned}
x_1 &= \normx \cos(\theta_{1}) \notag\\    
x_2 &= \normx \sin(\theta_{1})\cos(\theta_{2}) \notag\\
\,\,\,\vdots \\
x_{d-1} &= \normx \sin(\theta_{1})\cdots\sin(\theta_{d-2})\cos(\theta_{d-1}) \notag\\
x_{d} &= \normx \sin(\theta_{1})\cdots\sin(\theta_{d-2})\sin(\theta_{d-1}) \notag
\,.
\end{aligned}
\end{gather}
In Fig. \ref{fig:hwk2circuitexamples}, we illustrate an example of the resulting quantum circuit for $k=2$ and $n=6$ qubits (hence $d=15$). 
An explicit step-by-step sketch of the execution of Alg. \ref{alg:hw-enc} for the same example is presented in Tab. \ref{tab:example6choose2} in App. \ref{app:tables}. 
The following lemma provides the total $\CNOT$-gate count to implement $\LOAD_{B_{k}}(\vecx)$ using Alg. \ref{alg:hw-enc}. 

\begin{lemma}[Total $\CNOT$ cost of $\LOAD_{B_k}(\vecx)$ for $\vecx\in\mathbb{R}^{d}$]
\label{lemma:totalCNOTcount_real_dense}
Let $n \ge 2$ and $k \in [1, \, n/2]$ be integers, $d = \binom{n}{k}$, and $\vecx \in \mathbb{R}^{d}$. 
The HW-$k$ encoder $\LOAD_{B_{k}}(\vecx)$ generated by Alg. \ref{alg:hw-enc} can be implemented using a number $\order{k \, d}$ of $\CNOT$ gates, namely
\begin{align}\label{eq:totalCNOTcount_real_dense}
\begin{cases}
    2(n-1), & k=1\\
    (n-2)(3n-1), & k=2\\
    \frac{1}{3}(n-3)(5n^2-6n-2), & k=3\\
    \frac{1}{12}(n-4)(13n^3-58n^2+79n-42), & k=4\\
    \frac{1}{12}\sum_{\ell=1}^4 a_\ell(n-k)^\ell + b_\ell     
    & k\ge5\,
\end{cases}
\end{align}
with $a_{1} = 178$, $a_{2} = 239$, $a_{3} = 98$, $a_{4} = 13$, and $b_{\ell} \le \sum_{\ell = 5}^{k - 1} \binom{n - (k - \ell)}{\ell + 1} \big( 16 \ell - 6 \big)$.
\end{lemma}

\begin{proof}
The total $\CNOT$-gate count is given by the general expression in Eq. \eq{totalCNOTcount}, where $C_{\ell}$ is the cost of compiling a single $\ell$-controlled $R^{\bin}_{\tout}(\theta)$ gate. 
We calculate this cost in Lemma \ref{lemma:cgRBScompilationcomplex}, and the results are summarized in Tab. \ref{tab:cgRBScompilation} in App. \ref{app:compilation_one_gate}.
Plugging the results in Lemma \ref{lemma:cgRBScompilationcomplex} in Eq. \eq{totalCNOTcount} immediately leads to Eq. \eq{totalCNOTcount_real_dense}. 
\end{proof}

\begin{figure*}[!t]
\centering
\includegraphics[width=.815\textwidth]{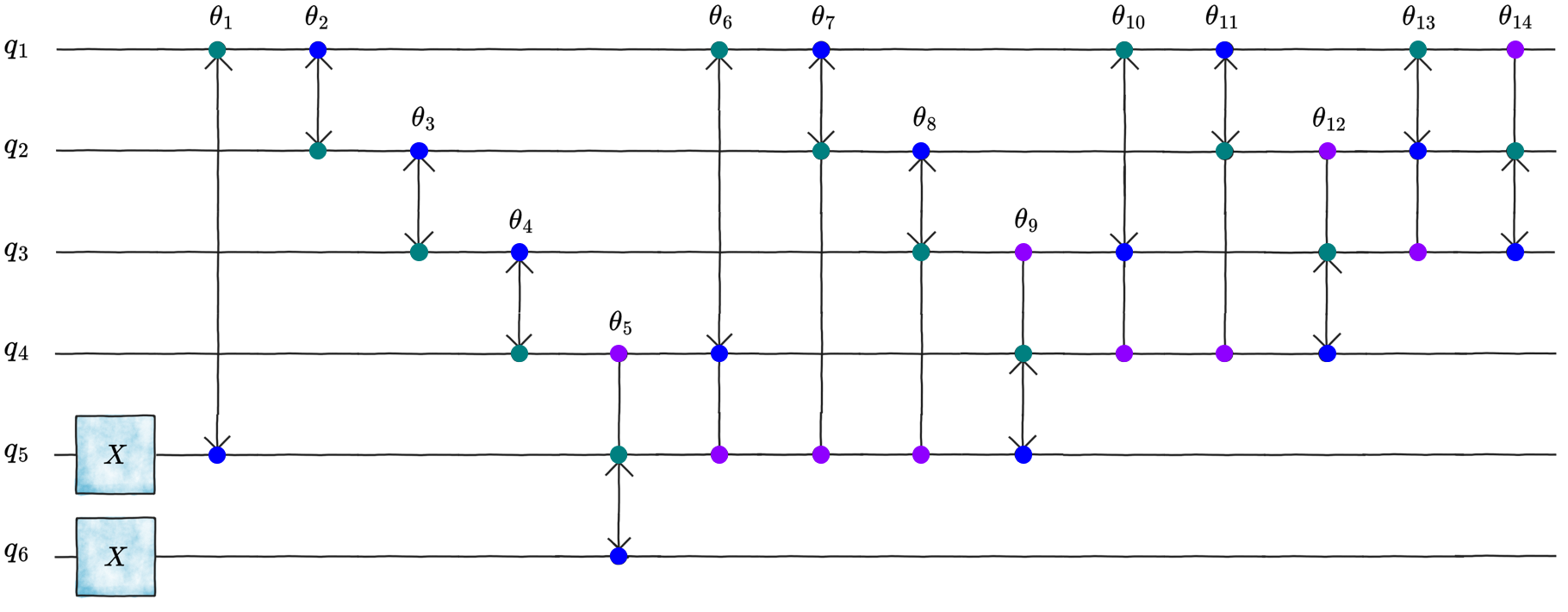}
\caption{
    \textbf{Fixed-Hamming-weight amplitude encoder.} Circuit generated by Alg. \ref{alg:hw-enc} with $n=6$ and $k = 2$ (see Tab. \ref{tab:example6choose2} for details). 
    Colors are used to indicate input (blue), output (green), and control (violet) qubits of each controlled $\RBS$ gate. 
    Angles $\{\theta_{j}\}_{j \in [14]}$ are calculated using Eq. \eq{spherical_coords}.
    }
\label{fig:hwk2circuitexamples}
\end{figure*}

\subsubsection{Dense complex-valued data}
\label{sec:complex_dense}

\begin{algorithm}[t!]
\DontPrintSemicolon
{\bf Input:} number of qubits $n$, Hamming weight $k$, and data vector $\vecx$ of dimension $d\le\binom{n}{k}$\;
{\bf Output:} Quantum circuit $\mathcal{C}=\LOAD_{B_k}(\vecx)$\;
\;
\caption{\label{alg:hw-enc-cplx} HW-$k$ encoder for dense $\vecx\in\mathbb{C}^d$}
\DontPrintSemicolon
\SetKwFunction{FMain}{\text{\sc{HWEncoder}}}
\SetKwFunction{FNextHWkBS}{\text{\sc{NextBS}}}
\SetKwFunction{GateParams}{\text{\sc{GateAddresses}}}
\SetKwProg{Fn}{Function}{:}
\Fn{\FMain{n, k, $\vecx$}}
{\;
    $\mathcal{C} \leftarrow {\rm Circuit}(n)$\;
    \For{$i \leftarrow n-k+1 \; \KwTo \; n$} {
        Add gate $X$ on qubit $i$ to $\mathcal{C}$
    } 
    $b \leftarrow 1^{k}0^{n-k}$ \Comment{initial bitstring}\; 
    $marked \leftarrow \{1,\ldots,k\}$ \Comment{indices of marked bits}\;
    $untouched \leftarrow \{n-k+1,\ldots,n\}$ \Comment{indices of ones that have never been flipped}\;
    \For{$j \leftarrow 1 \; \KwTo \; d-1$} {
        $\{b^{\prime},marked^\prime\} \leftarrow \FNextHWkBS(b, marked)$ \;
        $\{\bin,\tout,\vctrl,untouched^\prime\} \leftarrow \GateParams(b,b^{\prime},untouched)$ \;
        Calculate $\theta_j$ and $\phi_j$ via Eqs. \eq{cplx_arg_angles_theta} and \eq{cplx_arg_angles_phi} \;
        Add gate $c_{\vctrl}R^{\bin}_{\tout}(\theta_j,\phi_j)$ to $\mathcal{C}$\;
        $\{b,marked,untouched\} \leftarrow \{b^{\prime},marked^\prime,untouched^\prime\}$ \;
    }
    Calculate $\phi_d$ via Eq. \eq{cplx_arg_angles_phi} \;
    $\vctrl \leftarrow $ set of $k$ indices where the bits of $b'$ are 1  \;
    $\bin \leftarrow$ any index where $b'$ has a 0  \;
    Add gate $c_{\vctrl}\overline{\text{Ph}}(\phi_d)$ acting on qubit $\bin$ to $\mathcal{C}$\;
    }
    \KwRet{$\mathcal{C}$}
\end{algorithm}

For $\vecx \in \mathbb{C}^{d}$, we need two sets of angles, $\{\theta_{j}\}_{j \in [d-1]}$ and $\{\phi_j\}_{j \in [d]}$, to encode the absolute values and phases of each amplitude $\left\{x_{j} = \abs{x_{j}} \, \exp(i\arg(x_{j}))\right\}_{j \in [d]}$.

The $d - 1$ angles $\theta_{j}$ are responsible for the encoding of the absolute values $\abs{x_{j}}$, and are calculated as the hyperspherical coordinates of the vector $\{\abs{x_{1}}, \, \ldots, \, \abs{x_{d}} \}$, similarly to Eq. \eq{spherical_coords}. 
Explicitly, the $\theta_{j}$ angles are given by

\begin{gather}\label{eq:cplx_arg_angles_theta}
\begin{aligned}
\theta_{j} &\coloneqq \atantwo\left(\sqrt{\sum_{j^{\prime} = j + 1}^{d} \abs{x_{j^{\prime}}}^2}, \,\, \abs{x_{j}} \right),\,\, j \in [d-2] \, ; \\
\theta_{d-1} &\coloneqq \atantwo\big( \abs{x_{d}}, \, \abs{x_{d-1}} \big) \, .
\end{aligned}
\end{gather}
The first $d - 1$ angles $\phi_{j}$, on the other hand, are responsible for the encoding of the complex phases of each entry $\{x_{j}\}_{j \in [d-1]}$. 
The first angle, $\phi_{1}$, is directly related to the complex phase of the first data entry, $x_{1}$.
Inspecting Eq. \eq{complexRBS}, one can see that, while the complex $\RBS$ gate encodes a desired complex phase onto the amplitude of an \quotes{initial} computational basis state, it also encodes the complex conjugate of the same phase onto the amplitude of the \quotes{new} computational basis state created in the superposition.
This leads to the accumulation of undesired phases every time a complex $\RBS$ gate is applied.
To remove these undesired phases, we add correction terms to the subsequent angles, $\phi_{j > 1}$ as follows:

\begin{gather}\label{eq:cplx_arg_angles_phi}
\begin{aligned}
\phi_{1} &\coloneqq - \arg\big( x_{1} \big) \,, \\
\phi_{j} &\coloneqq - \arg\big(x_{j}\big) + \sum_{j^{\prime} = 1}^{j - 1} \, \phi_{j^{\prime}} \, , \quad j \in  [2, \, d] \, .
\end{aligned}
\end{gather}
To guarantee numerical stability, all the angles are taken $\mod 2\pi$. In addition to the complex $\RBS$ gate defined in Eq. \eq{complexRBS}, we also need to add an anti-phase gate $\overline{\text{Ph}}(\phi) \, \coloneqq \, e^{-i\phi} \, R_{z}(\phi/2)$ at the end of the circuit to cancel out the excess complex phase created by the sequence of complex RBS gates, where $\overline{\text{Ph}}(\phi)$ is controlled by all the qubits corresponding to a $1$ value in $b_{d-1}$ and acts on any qubit corresponding to a $0$ value.

Thus, the resulting quantum circuit has the same architecture of controlled-RBS gates as in the real case (see Fig. \ref{fig:hwk2circuitexamples}), plus an extra (controlled) $\overline{\text{Ph}}$ gate at the end of the circuit, as well as the choice of gate synthesis (see Fig. \ref{fig:rbscomplex}). 
The following lemma provides the total $\CNOT$-gate count to implement $\LOAD_{B_k}(\vecx)$ using Alg. \ref{alg:hw-enc-cplx}. 

\begin{lemma}[Total $\CNOT$-gate cost of $\LOAD_{B_k}(\vecx)$ for $\vecx\in\mathbb{C}^d$]
\label{lemma:totalCNOTcount_cpx_dense}
Let $n\ge2$ and $k\in[1,n/2]$ be integers, $d=\binom{n}{k}$, and $\vecx\in\mathbb{C}^d$. The HW-$k$ encoder $\LOAD_{B_k}(\vecx)$ generated by Alg. \ref{alg:hw-enc} can be implemented using a number $\mathcal{O}(k\,d)$ of $\CNOT$ gates, namely
\begin{align}\label{eq:totalCNOTcount_cpx_dense}
\begin{cases}
    2(n-1), & k=1\\
    (n-2)(3n-1), & k=2\\
    \frac{1}{3}(n-3)(7n^2-12n+2), & k=3\\
    \frac{1}{12}(n-4)(19n^3-86n^2+105n-30), & k=4\\
    \frac{1}{12}\sum_{\ell=1}^4 \tilde{a}_\ell(n-k)^\ell + \tilde{b}_\ell     
    & k\ge5\,
\end{cases}
\end{align}
with $\tilde{a}_1=230,\tilde{a}_2=329,\tilde{a}_3=142,\tilde{a}_4=19$, and $\tilde{b}_\ell\le\sum_{\ell=5}^{k-1} \binom{n-(k-\ell)}{\ell+1} \big(20\ell+4\big)$.
\end{lemma}

\begin{proof}
The total $\CNOT$-gate count is given by the general expression \eq{totalCNOTcount}, where $C_{\ell}$ here is the cost of compiling a single $\ell$-controlled $R^{\bin}_{\tout}(\theta, \, \phi\ne0)$ gate. This is computed in Lemma \ref{lemma:cgRBScompilationcomplex} in App. \ref{app:compilation_one_gate}, and plugging in the expression immediately leads to Eq. \eq{totalCNOTcount_cpx_dense}. 
\end{proof}

\section{Sparse encoder}
\label{sec:sparse}

Here, we extend the construction of Sec. \ref{sec:hamming_weight} to the case of sparse data. 
For a $d$-dimensional data vector $\vecx$ having $s \ll d$ non-zero entries, the HW-$k$ encoder of Sec. \ref{sec:hamming_weight} still uses $d - 1$ parameterized gates. 
However, in this case, a parameter-optimal amplitude encoder, as expressed in Def. \ref{def:parameter_optimal}, would require only $s - 1$ gate parameters.

\begin{algorithm}[b!]
\DontPrintSemicolon
{\bf Input:} number of qubits $n$ and tuple $y\coloneqq \{(x_i,b_i)\}_{i\in [s]}$ of non-zero data values $x_i$ and their addresses $b_i$\;
{\bf Output:} Quantum circuit $\mathcal{C}=s\text{-Load}(y)$\;
\;
\caption{\label{alg:sparse-hw-enc} Amplitude encoder for sparse data}
\DontPrintSemicolon
\SetKwFunction{FMain}{\text{\sc{SparseHWEncoder}}}
\SetKwProg{Fn}{Function}{:}
\Fn{\FMain{n, $y$}} 
{\;
    $bs \leftarrow b_1$ \Comment{address $b$ of the $1^{\text{st}}$ element of the $y$ tuple} \;
    $\mathcal{I} \leftarrow $ set of indices where the bits of $bs$ have value 1  \;
    $\mathcal{C} \leftarrow {\rm Circuit}(n)$\;
    \For{$i \leftarrow 1 \; \KwTo \; |\mathcal{I}|$} {
        Add gate $X$ on qubit $\mathcal{I}_i$ to $\mathcal{C}$\;
    } 
    $untouched \leftarrow [n]\setminus \mathcal{I}$ \;
    \For{$j \leftarrow 2 \; \KwTo \; s$} {
        $bs^\prime \leftarrow b_j$ \;
        $\{\mathcolor{blue}{in},\mathcolor{teal}{out},\mathcolor{violet}{ctrl},untouched^\prime\} \leftarrow \GateParams(bs,bs^\prime,untouched)$ \;
        Calculate $\theta_j$ and $\phi_j$ via Eqs. \eq{cplx_arg_angles_theta} and \eq{cplx_arg_angles_phi} \; \Comment{using  $x$ values from $y$} \;
        Add gate $c_{\vctrl}R^{\bin}_{\tout}(\theta_j,\phi_j)$ to $\mathcal{C}$\;
        $\{b,untouched\} \leftarrow \{b^{\prime},untouched^\prime\}$ \;
    }
    Calculate $\phi_d$ via Eq. \eq{cplx_arg_angles_phi} \;
    $\vctrl \leftarrow $ set of $k$ indices where the bits of $b'$ are 1  \;
    $\bin \leftarrow$ any index where $b'$ has a 0  \;
    Add gate $c_{\vctrl}\overline{\text{Ph}}(\phi_d)$ acting on qubit $\bin$ to $\mathcal{C}$\;
    }
    \KwRet{$\mathcal{C}$}
\end{algorithm}

In this setting, we consider a sparse-access model, where the data vector $\vecy$ to be encoded is of the form
\begin{align}\label{eq:sparse_vector}
\vecy \coloneqq \{(x_{1}, \, b_{1}), \, (x_{2}, \, b_{2}), \, \ldots, \, (x_{s}, \, b_{s}) \} \, , 
\end{align}
with $\{x_{j}\}_{j \in [s]}$ being the non-zero components of $\vecx$, and $\{b_{j}\}_{j \in [s]}$ being the addresses (in bitstring format) associated with these values. 
The goal is to prepare the state 
\begin{align}\label{eq:sparsestate}
s\text{-}\mathrm{Load}(\vecy) \, \ket{0}^{\otimes n} \coloneqq \frac{1}{\normx} \, \sum_{j \in [s]} \, x_{j} \, \ket{b_{j}} \, ,
\end{align}
similarly to the fixed HW encoders discussed in Sec. \ref{sec:hamming_weight}. 
However, in the sparse-access model, $b_{j}$ and $b_{j + 1}$ do not need to satisfy any constraints, \emph{e.g.} equal HWs or fixed Hamming distance. 
Consequently, the complex $\RBS$ gate used in the previous sections is not enough to cover all possible sparsity configurations.
To this end, we need the generalized $\RBS$ gate, $R_{\toutbf}^{\binbf}(\theta, \, \phi)$ in Eq. \eq{gR}.
With this $(m + m^{\prime})$-qubit gate, it is possible
to create superpositions of computational basis states of different Hamming weights.

\begin{figure}[t!]\label{fig:hwk3circuitexamples}
\includegraphics[width=.8\columnwidth]{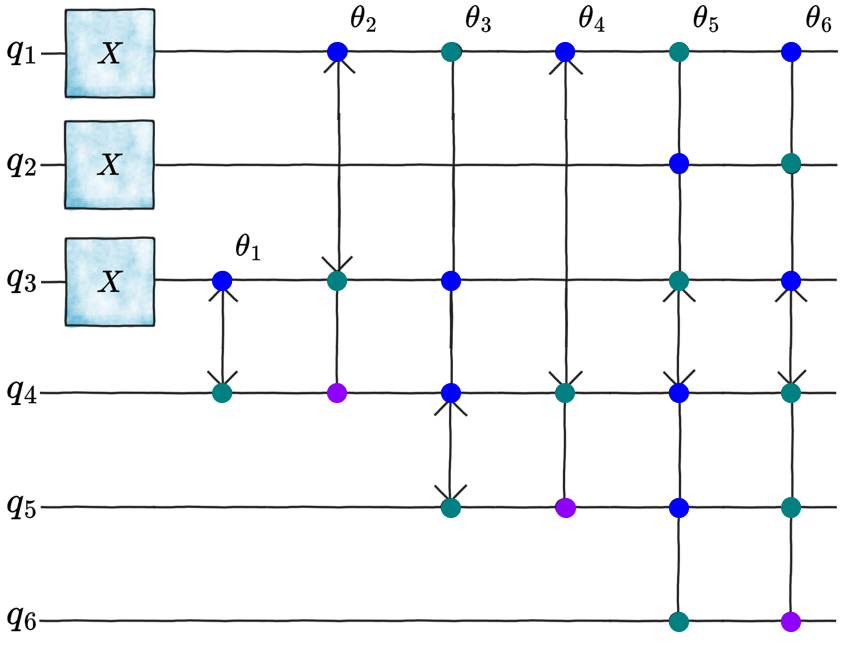}
\caption{\textbf{Sparse quantum data loader.} An example with $n=6$ qubits, and sparsity $s = 7$ with $x_j\ne0$ at positions $b_{j} \in \{7,11,14,19,26,37,58\}$ for real data. 
Colors are used to indicate the input (blue), output (green), and control (violet) qubits of each $\gRBS$ gate.
Angles are calculated using Eq. \eq{spherical_coords} for the non-zero entries. 
The execution of Alg. \ref{alg:sparse-hw-enc} to construct this circuit is shown in Tab. \ref{tab:example6choose3} in App. \ref{app:tables}. 
This circuit can be compiled using at most $174$ $\CNOTs$ (see Tab. \ref{tab:cgRBScompilation} in App. \ref{app:compilation_one_gate}).} 
\end{figure}

The procedure to build the circuit is given in Alg. \ref{alg:sparse-hw-enc} for complex data and goes as follows. Based on the first bitstring address $b_1$, we apply Pauli-$X$ gates to generate the initial state $\ket{b_1}$. 
Then, we use Alg. \ref{alg:rbs-gate-params} to extract the generalized $\RBS$ gate parameters (input, output, and control qubits) needed to generate the superposition between $\ket{b_{1}}$ and $\ket{b_{2}}$, and add the gate $R^{\mathcolor{blue}{in_1,\ldots,in_m}}_{\mathcolor{teal}{out_1, \ldots, out_{m'}}}(\theta_{1}, \, \phi_{1})$ to the circuit; 
the angles are computed from Eqs. \eq{cplx_arg_angles_theta} and \eq{cplx_arg_angles_phi}. 
The same procedure is repeated for all elements in the tuple in Eq. \eq{sparse_vector}, adding $s - 1$ (possibly controlled and generalized) $\RBS$ gates to the resulting circuit.
Alg. \ref{alg:rbs-gate-params} can output lists of inputs and outputs with different lengths (\emph{i.e.} $m\ne m'$), depending on whether an increase/decrease of Hamming-weight is needed to superpose $\ket{b_{j}}$ and $\ket{b_{j + 1}}$. 
The ancilla-free circuit architecture output by Alg. \ref{alg:sparse-hw-enc} will depend on the particular sparsity structure of the data vector $\vecy$. In the case of real-data encoding, we compute the angles $\theta_j$ from Eq. \eq{spherical_coords} instead of calculating $\theta_j$ and $\phi_j$ via Eqs. \eq{cplx_arg_angles_theta} and \eq{cplx_arg_angles_phi} in addition to not executing lines $18-21$ of Alg. \ref{alg:sparse-hw-enc}. 
The worst-case bound on the total number of $\CNOTs$ can be calculated using the compilation of $\gRBS$ gates presented in Tab. \ref{tab:cgRBScompilation} in App. \ref{app:compilation_one_gate}. 
An explicit example for $n = 6$ qubits and sparsity $s = 7$ is illustrated in Fig. \ref{fig:hwk3circuitexamples}, and in Tab. \ref{tab:example6choose3} in App. \ref{app:tables}. 

In Fig. \ref{fig:plot_appendix}(a) in App. \ref{app:cnot_cost_sparse_binary}, we show a numerical comparison between the encoders from Alg. \ref{alg:sparse-hw-enc}, Ref. \cite{Veras2022} (as implemented in the \texttt{qclib} package \cite{Araujo2023}), and Ref. \cite{Shende2006} (as implemented in the \texttt{Qiskit} package \cite{qiskit2024}).
We compare, as a function of the number of qubits $n$, the average $\CNOT$-gate count to encode $100$ random $2^n$-dimensional data vectors $\bf{y}$ of sparsity $s \in \{n, \, n^{2}\}$. 
For each instance, we used a Haar-random vector of non-vanishing amplitudes $\bf{x}$ and uniformly sampled computational basis state addresses $\ket{b_j}$. 
The results indicate that the average $\CNOT$-gate count of Alg. \ref{alg:sparse-hw-enc} and of Ref. \cite{Veras2022} scales as $\order{\operatorname{poly}(n)}$ for both sparsity levels $s$.
Meanwhile, the algorithm from Ref. \cite{Shende2006} scales as $\order{2^{n}}$ regardless of the sparsity $s$.
Moreover, Ref. \cite{Veras2022}'s encoder displays an improvement of $\order{n}$ when compared to the results for Alg. \ref{alg:sparse-hw-enc}.
However, while our sparse encoder is ancilla-free, the sparse encoder from Refs. \cite{Veras2022, Araujo2023} uses $\order{n}$ ancillas.

\section{Binary encoder}
\label{sec:binary}

\begin{figure}[t!]
\centering
\includegraphics[width=\columnwidth]{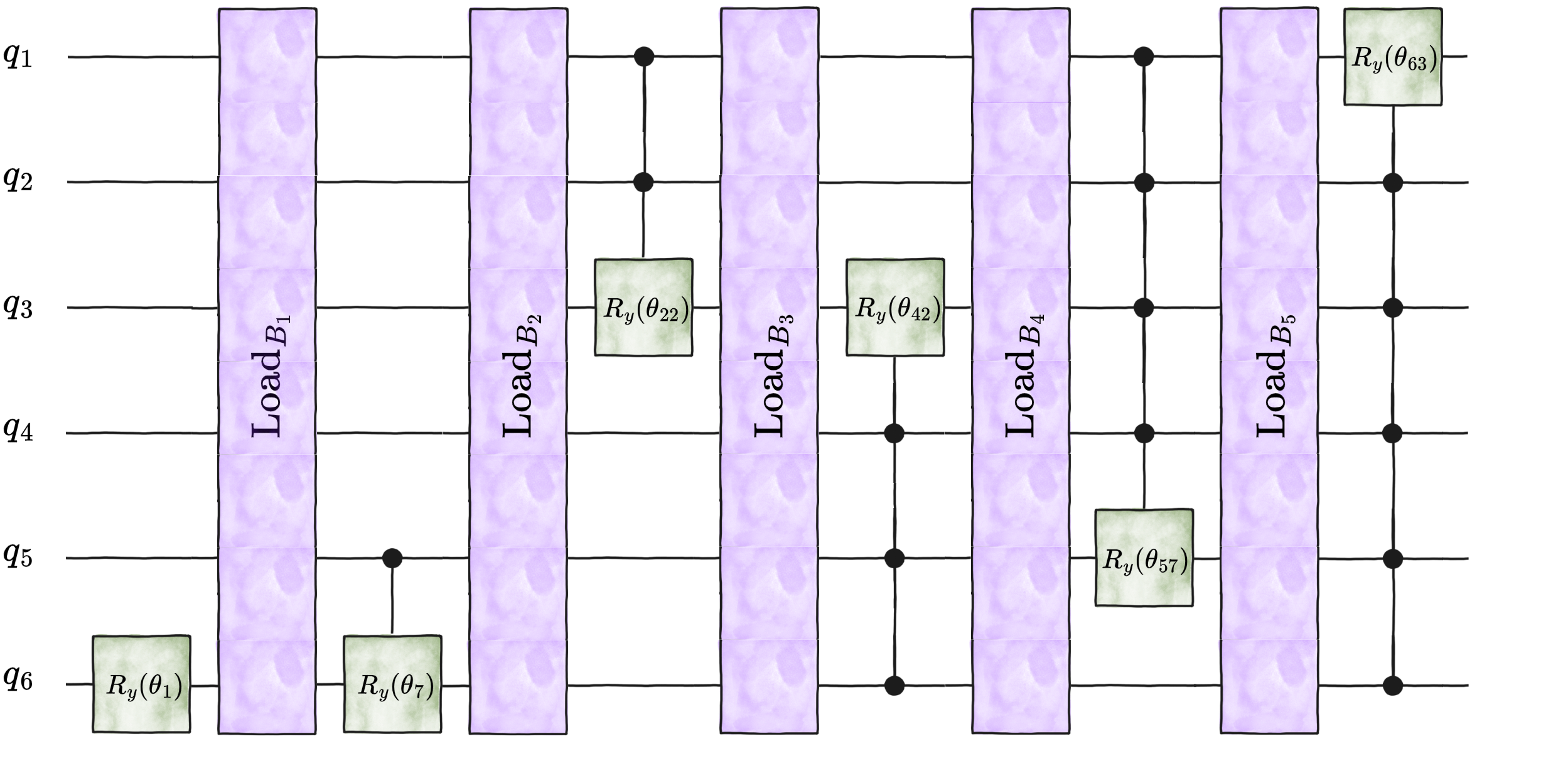}
\caption{
    \textbf{Binary-basis amplitude encoder based on HW-$k$ encoders.}
    Encoding of a real, $64$-dimensional vector $\vecx$ on a $6$-qubit quantum state.
    The initial state is the zero state, $\ket{000000}$, and the circuit is constructed by stacking HW-$k$ encoders, $\LOAD_{B_{k}}$, with $k \in [0, \, 6]$, and (controlled) $R_{y}$ rotations.
    An initial $R_{y}(\theta_{1})$ gate creates a superposition between $\ket{000000}$ and $\ket{100000}$.
    Subsequently, the bitstring $1\,0^{5}$ serves as the input of the Ehrlich Algorithm \ref{alg:nxt-hw-bs} that generates the circuit $\LOAD_{B_{1}}$, responsible for the encoding of amplitudes associated to all $6$ computational basis states of HW-$1$.
    Next, based on the last computational basis state added to the superposition by the previous $\LOAD_{B_{1}}$, the controlled $R_{y}(\theta_{7})$ rotation adds the first computational basis state of HW-$2$.
    The bitstring associated with this state is then used as input for the next Ehrlich Algorithm.
    The procedure is repeated until the last amplitude is encoded in the $\ket{111111}$ state.
    The gate parameters $\{\theta_{j}\}_{j \in [63]}$ are calculated using  Eq. \eq{spherical_coords}. 
    See Tab.  \ref{tab:hwktobinary} in App. \ref{app:tables} for more details. 
    This circuit can be compiled using at most $1048$ $\CNOTs$ (see Tab. \ref{tab:cgRBScompilation} in App. \ref{app:compilation_one_gate}).
} 
\label{fig:hwtobincircuitexample}
\end{figure}

In this section we show how to stack together the dense HW-$k$ encoders introduced in Sec. \ref{sec:hamming_weight} to encode data on multiple constant-HW subspaces (e.g., $\LOAD_{B_{k_{1}}\cup B_{k_{2}}}$ with $k_{1} < k_{2}$). 
The strategy is to populate the different fixed HW subspaces in ascending order while using a generalized $\RBS$ gate to connect the two subsequent subspaces. 
For instance, $\LOAD_{B_{k_{1}}\cup B_{k_{2}}}$ can be built as $\LOAD_{B_{k_{1}}}$ followed by a $\gRBS$ gate that increases $k_{1}$ to $k_{2}$ and $\LOAD_{B_{k_{2}}}$; 
in the special case of successive HWs, \emph{i.e.} $k_{2} = k_{1} + 1$, the $\gRBS$ can be replaced by a $k_{1}$-controlled $R_{y}$ gate.

In particular, this strategy can be used to build a full binary-basis amplitude encoder by stacking dense HW-$k$ encoders of all possible $0 \leq k \leq n $ in ascending order. 
We will first explain the general ideal for real-valued data and will later generalize it to complex data.
The procedure goes as follows: 
($i$) first, initialize the circuit in the state $\ket{0^{n}}$; 
($ii$) apply a $R_{y}$ gate in the last qubit to generate the superposition $\cos\theta_1\ket{0^{n}}+\sin\theta_1\ket{10^{n-1}}$ according to Eq. \eq{RY}; 
($iii$) use $10^{n-1}$ as the initial bitstring of the Ehrlich algorithm \ref{alg:nxt-hw-bs}, creating the circuit $\LOAD_{B_1}$, which is responsible for adding all HW-$1$ computational basis states to the superposition; 
($iv$) apply a controlled $R_{y}$ rotation to add the first computational basis state of HW-$2$ to the superposition;
($v$) use the state encoded in the previous step as input of the algorithm that will generate $\LOAD_{B_{2}}$;
($vi$) iterate over steps ($iv$) and ($v$) until all HWs are populated, and the last multi-controlled $R_{y}$ generates the superposition with $\ket{1^{n}}$. 
The (multi-)controlled $R_{y}$ gates applied between $\LOAD_{B_{k}}$ and $\LOAD_{B_{k+1}}$ should be chosen such that the last computational basis state added to the superposition is associated to a bitstring that is a viable input for the Ehrlich algorithm \cite[Theorem 2.4]{Even1973}. 
Alternatively, we can use the Hamming-weight-increasing step ($iv$) using a HW-mixing $\gRBS$ gate $R_{\toutbf}^{\binbf}(\theta, \, \phi)$ introduced in Eq. \eq{gR}.
However, since the operations necessary for the binary encoder are local, we chose to use the cheaper multi-controlled $R_{y}$ gates (see Tab. \ref{tab:cgRBScompilation} in App. \ref{app:compilation_one_gate}).

In case of complex amplitudes, the (controlled) $R_{y}(\theta_{j})$ rotations must be replaced by (controlled) $\text{Ph}(\phi_{j}) \, R_{y}(\theta_{j})$ to generate the superposition with the correct phases.
The angles $\{\theta_{j}\}_{j \in [d-1]}$ and $\{\phi_{j}\}_{j \in [d-2]}$ are calculated using Eqs. \eq{cplx_arg_angles_theta} and \eq{cplx_arg_angles_phi}.
The last multi-controlled $R_{y}$ gate, however, must be replace by a multi-controlled $\operatorname{SU}(2)$ rotation that can be decomposed as $R_{z}(\phi_{d-1}) \, R_{y}(\theta_{d-1}) \, R_{z}(\lambda)$.
The angles $\phi_{d-1}$ and $\lambda$ are calculated as
\begin{gather}
\begin{aligned}
    \phi_{d-1} &= \frac{1}{2}\left( \arg(x_{d}) - \arg(x_{d-1}) \right) \,; \\
    \lambda & = -\frac{1}{2} \left(\arg(x_{d}) + \arg(x_{d-1})\right) + \sum_{j=1}^{d-3} \, \phi_{j} \,.
\end{aligned}
\label{eq:last_angles_complex}
\end{gather}

An example of the binary encoder for $n = 6$ qubits is shown in Fig. \ref{fig:hwtobincircuitexample}. 
The total $\CNOT$-gate count is given in the following lemma.

\begin{lemma}[$\CNOT$ count of the binary encoder]
\label{lemma:binary}
Our binary encoder for $\vecx\in\mathbb{R}^{d}$ uses
\begin{align}\label{eq:cnotcountbinary}
\# \text{ $\CNOTs$} 
&\le \frac{1}{12}(13 n^{4} - 58 n^{3} + 119 n^{2} - 50 n + 120)\notag\\
&\quad+\sum_{k = 5}^{n - 1}\left[\binom{n}{k}(16k - 22) - 2\right].
\end{align}
\end{lemma}

\begin{proof}
First, it is important to notice that the two \quotes{tricks} mentioned in Sec. \ref{sec:hamming_weight} to reduce the number of controls (\emph{i.e.} eliminating initial controls and constructing a HW-($n-k$) encoder from the corresponding HW-$k$ encoder) cannot be used here, since the initial state for each $\LOAD_{B_k}$ in Fig. \ref{fig:hwtobincircuitexample} contains a superposition of states. Therefore, the cost of each $\LOAD_{B_k}$ is given by the upper bound $\left[\binom{n}{k}-1\right]C_{k-1}$ in the general expression in Eq. \eq{totalCNOTcount}, where $C_{k-1}$ is the cost of compiling a single ($k-1$)-controlled-$R^{\mathcolor{blue}{in}}_{\mathcolor{teal}{out}}(\theta)$ gate. 
After each $\LOAD_{B_k}$ is applied, one needs a $k$-controlled $R_{y}$ gate to increase the HW to $k+1$. 
The total $\CNOT$ count of the circuit is therefore $\sum_{k=0}^{n-1}\left\{\left[\binom{n}{k}-1\right]C_{k-1}+\chi_k\right\}$, where $\chi_k$ is the $\CNOT$ cost of a $k$-controlled $R_{y}$. 
Substituting the values of $C_{k-1}$ and $\chi_k$ calculated in App. \ref{app:compilation_one_gate} immediately leads to Eq. \eq{cnotcountbinary} See Tab. \ref{tab:cgRBScompilation} for a summary, and Lemmas \ref{lemma:cRYcompilation} and \ref{lemma:cgRBScompilationcomplex} for the derivation. 
The first line contains exact $\CNOT$ counts coming from $k=1$ to $4$, while the second line contains upper bounds coming from $k\ge5$. 
\end{proof}

Numerically, we observe the $\CNOT$-gate count of Eq. \eq{cnotcountbinary} to be $\le 8 \, n \, 2^{n}$ (see Fig. \ref{fig:plot_appendix}(b) in App. \ref{app:cnot_cost_sparse_binary}). 
Asymptotically, this is the same scaling obtained by the sparse encoder from Ref.~\cite{Veras2022}. 
While the $\CNOT$-gate count of Ref. \cite{Plesch2011} is $\mathcal{O}(2^{n})$, they perform their amplitude encoding in the Schmidt basis with at most $2^{n/2}$ coefficients.
In contrast, we explore all $2^{n}$ computational basis states.

\section{Experimental and numerical results}
\label{sec:results}

In this Section, we show a proof-of-principle deployment of the dense HW-$k$ encoder on the \ionq's \texttt{Aria-1} quantum processor.
We also numerically demonstrate that the encoder circuits can be used as an ansatze for a variational quantum algorithm.
Then, we analyze the fidelity of the encoder under circuit noise.

\subsection{Quantum hardware demonstration}
\label{sec:hardware}

To demonstrate our encoding protocol, we encode a $q$-Gaussian probability distribution on the \texttt{Aria-1} quantum processor from \ionq \cite{IONQ}. 
The $q$-Gaussian probability density function $p_{q,\beta}(x)$, where $q \in (-\infty,\, 3)$ and $\beta \in (0,\, \infty)$, is proportional to the \emph{$q$-exponential} function $e_{q}\left(-\beta \, x^{2}\right)$, defined as  $e^{x} \,$ if $q = 1$; $\,[1 + (1 - q)x]^{1/(1 - q)}\,$ if $q \neq 1$ and $1 + (1-q)x > 0$; and $0^{1/(1-q)}\,$ otherwise \cite{tsallis1988, Tsallis1995, Prado1999}.  
This family of distributions has a wide range of applications, \textit{e.g.} nonextensive statistical mechanics \cite{Tsallis2009}, finance \cite{Borland2002, Borland2004}, metrology \cite{Witkovsky2023}, and biology \cite{Navarro2011}.
It is worth noting that in general the $q$-Gaussian is a non-log-concave distribution, falling outside of the scope of other encoding strategies \cite{Grover2002}. Here, we chose the parameters $q = 3/2$ and $\beta = 2$. 

\begin{figure}[t!]
    \centering
    \includegraphics[width=\columnwidth]{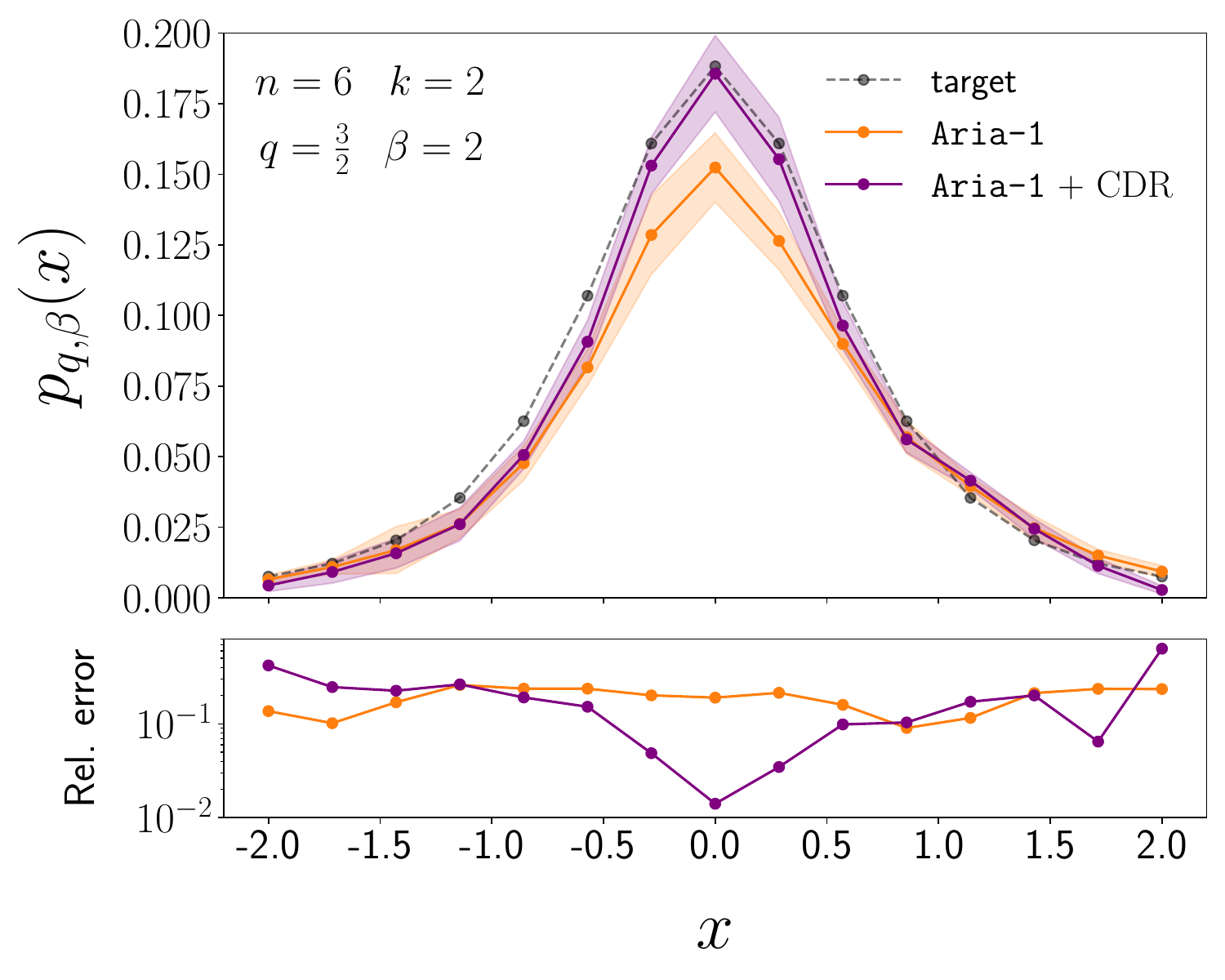}
    \caption{
        \textbf{Experimental deployment on \ionq.}
        \emph{(Top)} Experimental implementation of the encoding of a $q$-Gaussian probability distribution $p_{q,\beta}(x)$ on \ionq's \texttt{Aria-1} processor.
        The domain of $p_{q,\beta}(x)$ was truncated to $x \in [-2, 2]$ and discretized into $d=15$ points (dashed black line). 
        The circuit deployed used $n=6$ and $k=2$ and is shown in Fig. \ref{fig:hwk2circuitexamples}. 
        The solid purple (orange) line represents experimental results with (without) error mitigation using Clifford Data Regression ($\mathrm{CDR}$). 
        Shaded areas represent the uncertainty regions estimated via bootstrapping.
        \emph{(Bottom)} Relative error of experimental implementation of $p_{q,\beta}(x)$ w.r.t. the target function, in log scale. 
    }
    \label{fig:ionq}
\end{figure}

We use the circuit in Fig. \ref{fig:hwk2circuitexamples} for $6$ qubits and Hamming weight $k = 2$, which allows us to encode a data vector of size $d = \binom{6}{2} = 15$. 
This vector corresponds to a discretization of the distribution truncated to the interval $x \in [-2, 2]$. 
For this circuit, the first $4$ parameterized rotations are uncontrolled RBS gates, while the last $10$ $\RBS$s require $1$ control each.
Each $\RBS$ gate is compiled using $2$ $\CNOT$ gates, while each controlled $\RBS$ requires $6$ $\CNOTs$, leading to a circuit with $68$ $\CNOTs$ in total (see Tab. \ref{tab:cgRBScompilation} in App. \ref{app:compilation_one_gate}).
However, the \texttt{Aria-1} processor uses the M{\o}lmer-S{\o}rensen ($\operatorname{MS}$) gate as the two-qubit native gate \cite{IONQ}.
In this case, it is possible to implement one-controlled $R_{y}$ rotations using only one $\operatorname{MS}$ gate.
This allows us to implement the aforementioned circuit using $48$ $\operatorname{MS}$ gates.

In Fig. \ref{fig:ionq}, we show the results of the experimental implementation.
In the top panel, we plot the probability distribution estimated from the experiment by running the circuit $10^{4}$ times and measuring on the computational basis on each qubit to recover the encoded data. 
The solid orange line shows the empirical probability distribution estimated from the raw experimental data from the \texttt{Aria-1} quantum processor.
The solid purple line shows the estimated probability distribution after error mitigation using \emph{Clifford Data Regression} (\textrm{CDR}) \cite{czarnik2021, Lowe2021}.
 \textrm{CDR} models hardware noise by simulating near-Clifford circuits similar to the one at hand and then applying the trained response function to raw experimental data. 
 The training of the response functions was performed using \textrm{IonQ}'s capabilities for classical simulation of noisy circuits (see App.~\ref{app:cdr} for details).
The dashed black line shows the ideal target distribution. 
All relative errors concerning the target distribution are plotted in Fig.~\ref{fig:ionq} \emph{(Bottom)} in log scale.
We see that the data recovered from the experiment differs from the target distribution by $10$-$30\,\%$ in relative error.
After \textrm{CDR} is applied to the experimental data, relative errors can be improved.
For instance, the \textrm{CDR} protocol implemented improved the relative error at the peak of the distribution by one order of magnitude. 
The improvement can also be seen in the overall state infidelity $1-\Tr\left(\rho\,\sigma\right)$ ($\sigma$ being the state prepared by the noisy circuit), which decreased from $0.17$ to $0.08$ after error mitigation.
We attribute these discrepancies in the experimental results to the fact that \texttt{Aria-1} is a noise intermediate-scale quantum (\nisq) device.
The noise present in the circuit creates amplitudes \quotes{outside} of the subspace of interest.
Consequently, there is a reduction in the probability density \quotes{inside} said subspace.
We showed that this effect can be reduced by the application of error mitigation techniques such as \textrm{CDR}.

\subsection{HW-$k$ encoder under local depolarizing noise}
\label{sec:noise}

\begin{figure}[t!]
    \centering
    \includegraphics[width=\columnwidth]{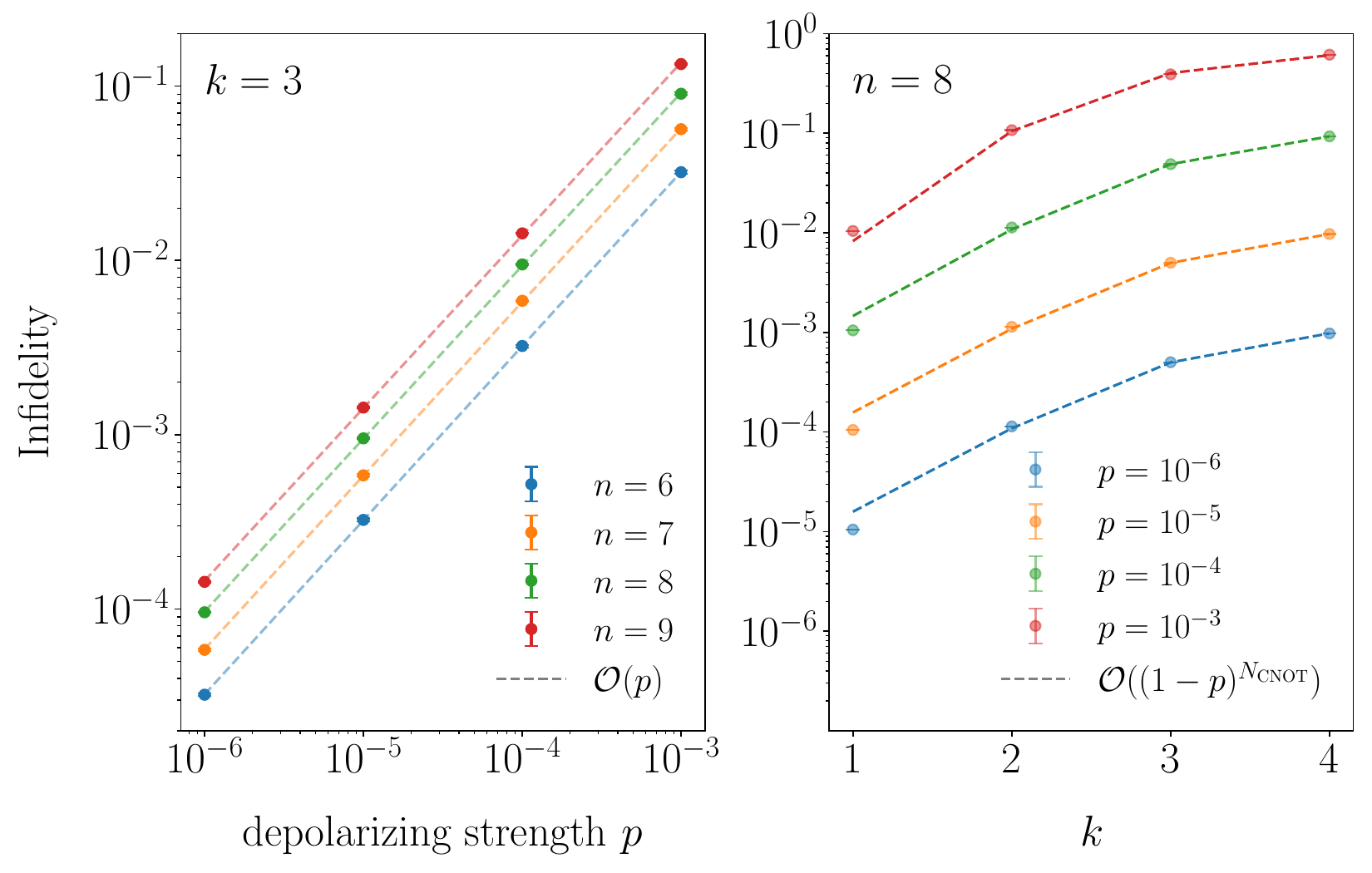}
    \caption{
        \textbf{Performance of HW-$k$ encoder under local depolarizing noise model.}
        Circuits for real-valued data are compiled into $\CNOTs$ and single-qubit gates (see Secs. \ref{sec:preliminaries} and \ref{sec:real_dense}).
        Each $\CNOT$ is followed by a two-qubit depolarizing channel of strength $p$.
        (\emph{Left}) State infidelity as a function of depolarizing strength $p$, averaged over $10$ Haar-random instances, for fixed $k = 3$ and $n\in[5,9]$. Dashed lines represent a linear fit $\order{p}$. 
        (\emph{Right}) State infidelity as a function of $k$, average over $10$ Haar-random instances for fixed $p \in \{10^{-3}, \, 10^{-4}, \, 10^{-5}, \, 10^{-6} \}$ and $n = 8$.
        Dashed lines represent curve fit of the function $\order{(1-p)^{N_{\CNOT}}}$ to the numerical data, where $N_{\CNOT} \equiv N_{\CNOT}(n, \, k) = \mathcal{O}\left(k \, \binom{n}{k}\right)$ is the number of $\CNOTs$ in each circuit for a given $n$ and $k$.
    }
    \label{fig:robustness}
\end{figure}

Motivated by the results in the Sec. \ref{sec:hardware}, in this section we numerically investigate the robustness of the HW-$k$ encoder in the presence of circuit noise as a function of both $k$ and the noise strength. 
For the noise model, given a pure quantum state $\rho$ and any pair of qubits labelled by $A$, we assume every $\CNOT$ gate acting on $A$ is followed by a depolarizing channel $\mathcal{D}_p$,
\begin{align}
    \mathcal{D}_{p}(\rho) = (1 - p) \, \rho + p \, \Tr_{A}(\rho) \otimes \frac{\mathbbm{1}}{4} \, ,
    \label{eq:depolarizing_channel}
\end{align}
where $p$ is the depolarizing strength, $\mathbbm{1}$ is the two-qubit identity matrix, and $\Tr_{A}(\cdot)$ is the partial trace over $A$. 
This noise model is compatible with reported noise models derived from randomized benchmarking techniques \cite{Helsen2022}.
Since two-qubit gate errors are the main source of circuit noise in \nisq devices, we assume that all single-qubit gates are noiseless.
We compiled the circuits resulting from Alg. \ref{alg:hw-enc} into $\CNOTs$, single-qubit gates, and multi-controlled $R_{y}$ rotations as described in Secs. \ref{sec:preliminaries} and \ref{sec:real_dense}, and App. \ref{app:compilation_one_gate}.
While the circuits were simulated using the \texttt{Qibo} package \cite{qibo2021}, the multi-controlled $R_{y}$ gates were compiled into $\CNOTs$ and single-qubit gates using the integration between the \texttt{qclib} library \cite{Araujo2023} and the \texttt{Qiskit} package \cite{qiskit2024}.

In Fig. \ref{fig:robustness} (left panel), we show a $\log$-$\log$ plot of the resulting infidelity for fixed $k = 3$ as a function of the depolarizing strength $p$ for $n \in [5, \, 9]$. 
We observed that for fixed $n$ and $k$, the average infidelity scales polynomially with $p$, namely $\order{p^\gamma}$ with $\gamma\approx1$. 
For $p \sim 10^{-3}$, which is a per-gate noise level compatible with several of the currently available quantum platforms \cite{IONQCloud, IBMQuantum, Quantinuum, Willow2024}, the infidelities in the range of $n$ studied are in between $10^{-2}$ and $10^{-1}$. 
In Fig. \ref{fig:robustness} (right panel), we plot the infidelity (in $\log$ scale) as a function of $k$ for fixed $n = 8$, and $p \in \{10^{-6}, \, 10^{-5}, \, 10^{-4}, \, 10^{-3} \}$. 
We observed that for fixed $n$ and $p$, the average infidelity scales exponentially with $k$ as $\order{(1-p)^{N_{\CNOT}(n,k)}}$, where $N_{\CNOT}(n,k)$ is the total CNOT gate count given by Theorem \ref{lemma:totalCNOTcount_real_dense}. This scaling can be intuitively explained from the fact that the depolarizing channel is isotropic, namely, the composition of $N_{\CNOT}$ local depolarizing channels is equivalent to a global depolarizing channel with strength $1-(1-p)^{N_{\CNOT}}$.
As expected from \nisq devices, the accumulation of errors on deep circuits leads to the preparation of states with high infidelity concerning the target states.
However, as reflected by the linear scaling in $p$, we see that each order-of-magnitude reduction in $p$ yields also an order-of-magnitude reduction in infidelity.

\subsection{HW-$k$ encoder as a variational quantum ansatz}
\label{sec:numerical}

\begin{figure}[t!]
    \centering
    \includegraphics[width=\columnwidth]{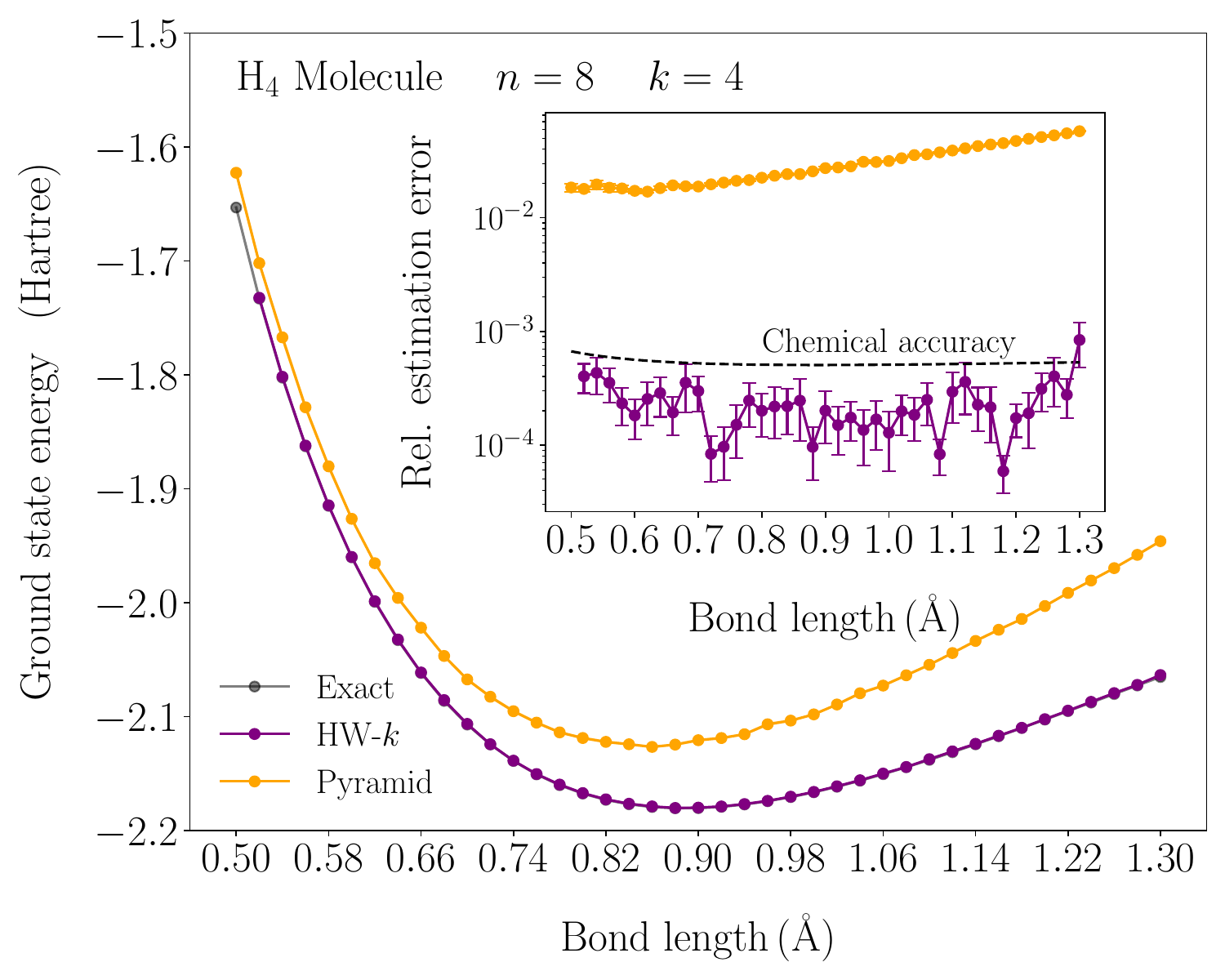}
    \caption{
        \textbf{Hamming-weight-$k$ encoder as a variational quantum ansatz.}
        (\emph{main}) Ground state energy (in Hartree) for the $\operatorname{H_{4}}$ molecule as a function of the bond length (in Angstrom).
        The purple curve represents the energies obtained using Alg. \ref{alg:hw-enc} with $n = 8$ and $k = 4$, while the orange curve was obtained using the \quotes{pyramid} circuit architecture from Ref. \cite{Cherrat2024}.
        The number of gradient descent epochs is fixed in both cases to $1000$, results are averaged over $50$ random initializations, and error bars are smaller than the markers.
        The \emph{exact} energy curve calculated from brute-force diagonalization is plotted in gray and is nearly indistinguishable from the purple curve.
        (\emph{inset}) Energy estimation error relative to the exact ground state energy as a function of the bond length.
        The color scheme is the same as in the main plot.
        Relative chemical accuracy is plotted as a dashed black line.
    }
    \label{fig:vqe}
\end{figure}

Here, we demonstrate the use of the encoders presented in Secs. \ref{sec:hamming_weight}-\ref{sec:binary} as a variational quantum ansatz.
We focus on the HW-$k$ encoder of dense real-valued data, though results immediately extend to complex-valued or sparse data as well as binary encoders.

We first recall that, given $n$ and $k$, the quantum circuits generated by Alg. \ref{alg:hw-enc} parameterize the subspace spanned by the basis $B_{k}$ of size $d = \binom{n}{k}$ defined in Eq. \eqref{eq:HWk_basis}. 
Therefore, this parametrization is ideally suited for solving optimization tasks that have some type of Hamming weight constraints.
For instance, this is the case for ground-state energy estimation of molecular Hamiltonians \cite{Anselmetti2021, Arrazola2022}, which have particle number-preserving symmetry, as well as portfolio optimization constrained by a fixed number of assets \cite{He2023}. 
Without loss of generality, here we illustrate this feature using a quantum chemistry example.
Given a molecule having $k$ electrons and a basis set describing its orbitals, the Jordan-Wigner transform maps its fermionic Hamiltonian into a qubit Hamiltonian $H$ where each qubit corresponds to one of the $n$ molecular spin-orbitals and the state $\ket{1}$ ($\ket{0}$) is associated with an occupied (unoccupied) spin-orbital. 
The ground state of the qubitized molecular Hamiltonian is therefore constrained to a subspace $B_{k}$ of the form \eqref{eq:HWk_basis}, and one can use the corresponding HW-$k$ circuit $\LOAD_{B_k}$ as a variational ansatz to approximate the ground state by classically optimizing the $\RBS$ gate angles $\bm{\theta} \coloneqq \{\theta_j\}_{j\in[d-1]}$ to minimize the energy function 
\begin{align}
E(\bm{\theta}) \coloneqq \bra{0^{n}} \, \LOAD_{B_{k}}^{\dagger}\!(\bm{\theta}) \, H \, \LOAD_{B_{k}}\!(\bm{\theta}) \, \ket{0^{n}} \, .
\label{eq:loss_function}
\end{align} 

Below we numerically demonstrate this by finding the ground state energy of the $\operatorname{H_{4}}$ molecule, with the qubitized Hamiltonian obtained via the quantum dataset from \texttt{Pennylane} \cite{Pennylane2022, Utkarsh2023}.
Using the $\operatorname{STO-3G}$ basis set, this molecule can be described by a $n=8$ qubit Hamiltonian (corresponding to $8$ spin-orbitals) with $k=4$ electrons in the ground state.
The results are presented in Fig. \ref{fig:vqe}.
In the main panel, we plot the ground state energy (in Hartree) of the $\operatorname{H_{4}}$ molecule as a function of the bond length (in Angstrom) between the Hydrogen atoms.
We estimate the energies using our amplitude encoder $\LOAD_{B_4}$ as an ansatz (solid purple line) and compare our results with the \quotes{pyramid} HW-preserving ansatz from Ref. \cite{Cherrat2023} (solid orange line), composed solely of uncontrolled $\RBS$ gates. 
For a fair comparison, we allow this ansatz to have the same number of parameters as ours, \emph{i.e.} $\binom{n}{k} - 1 = 69$ parameters. 
We did not observe any improvement in expressivity of the pyramid ansatz by adding more parameters beyond this number, hence we did not compare the two ansatze by equating their number of two-qubit gates. 
We leveraged the integration between \texttt{Qibo} and \texttt{PyTorch} \cite{Pytorch2019} and used the $\operatorname{ADAM}$ optimizer \cite{Kingma2017} with fixed learning rate $\eta = 10^{-3}$ as the gradient descent method of choice for both ansatze.
We plot the best energy results after $1000$ epochs, averaged over $50$ Haar-random initializations.
The \quotes{exact} ground state energies, calculated via brute-force diagonalization of the corresponding Hamiltonians, are plotted in a solid black line.
We see that our encoder approximates very well the ground state energy of the $\operatorname{H_{4}}$ molecule throughout the entire energy curve, unlike the pyramid ansatz, whose results worsen with increased bond length. 
The \emph{inset} of Fig. \ref{fig:vqe} shows the energy estimation errors relative to the exact ground state energies. 
We see that, after $1000$ epochs, the HW-$k$ ansatz on average reached chemical accuracy while the pyramid ansatz remained at least an order of magnitude away throughout the entire bond length domain.

\section{Conclusions} 
\label{sec:conclusions}

We provide an efficient and explicit classical algorithm to construct, gate by gate, a quantum circuit that uploads an arbitrary data vector into a subspace of fixed Hamming-weight quantum states. 
The quantum circuit uses the minimum number of parameterized gates needed to express generic data of a given dimension. 
The construction uses (generalized) RBS gates and allows us to precisely state the quantum resources needed for its execution, as well as deploy a proof-of-concept instance on real quantum hardware. 
We also provide the tools needed to further explore quantum encoders for this subspace and beyond.

As shown for the sparse and binary encoder, our HW-$k$ encoder can be used as a subroutine to power more complex algorithms. 
These two examples are but a small sample, and future work will explore more avenues. 
Still, we remark on the importance of binary encoding, as most algorithms would benefit from robust state preparation work on this basis. 
Other encoding schemes, such as the unary-basis encoder, can deploy basis change circuits \cite{RamosCalderer2022} to allow for further post-processing. 
A generalization of this circuit to the HW-$k$ basis would open new possibilities for more efficient general state preparation. 

Furthermore, our algorithm brings interesting implications from the lens of quantum machine learning. 
Variational ansatze that explore a constrained space are popular due to their ability to mitigate barren plateaus \cite{Monbroussou2023}. 
As shown by the numerical experiment with the $\operatorname{H_{4}}$ molecule, our HW-$k$ encoder, due to its expressivity, has the potential to improve performance of optimizations that involve particle-preserving symmetry.
Further investigation of this potential is needed and we leave it as an open question for future work.
Algorithms with space compression that goes beyond linear, as is the case with our HW-$k$ encoder, are also a promising direction for other machine learning schemes \cite{Sciorilli2024}, and show how exploring these subspaces can be successful. 

Lastly, we believe that explicit constructive algorithms for state preparation, with precise analysis of the quantum resources required, are essential to reach useful quantum advantage. 

\paragraph{Note added.} 
\label{sec:note_added}

We note that during the completion of this paper an independent work \cite{Raveh2024deterministic} appeared within the context of Bethe state preparation for integrable spin chains. 
Their results have partial overlap with our HW-$k$ encoders for dense data presented in Sec. \ref{sec:hamming_weight}.

\section{Acknowledgments} 
\label{sec:acknowledgments}

We thank Ariel Bendersky and Andr\'{e} J. Ferreira-Martins for insightful discussions.
We thank Jadwiga Wilkens for the availability of the \emph{Quantum Circuit Library} \cite{Wilkens2023}.

\bibliography{references}

\appendix


\section{Gate compilations and $\CNOT$-gate counts}

\subsection{Compilation of multi-controlled gates}
\label{app:compilation_one_gate}

\begin{table*}[t!]
\begin{tabular}{|c|c|c|c|c|c|c|}\hline
     \rule{0pt}{2ex}
     \rule[-1.5ex]{0pt}{0pt}
\diagbox[width=8.em]{\# controls}{Gate}
      & $X$
      & $R_{y}(\theta)$
      & $R^{\bin}_{\tout}(\theta)$
      & $R^{\bin}_{\tout}(\theta, \, \phi)$
      & $R^{\mathcolor{blue}{in_{1}, \ldots, in_{m}}}_{\mathcolor{teal}{out_{1}, \ldots, out_{m^{\prime}}}}(\theta)$
      & $R^{\mathcolor{blue}{in_{1}, \ldots, in_{m}}}_{\mathcolor{teal}{out_{1},\ldots, out_{m^{\prime}}}}(\theta, \, \phi)$\\ \hline
           \rule{0pt}{3ex}
     \rule[-1.5ex]{0pt}{0pt}
$\ell = 0$   & $0$ & $0$      & $2$      & $2$     & $\le 18(m + m^{\prime}) - 42$   & $\le 22 (m + m^{\prime}) - 20$   \\ \hline
     \rule{0pt}{3ex}
     \rule[-1.5ex]{0pt}{0pt}
$\ell = 1$   & $1$ & $2$     & $6$     & $6$     & $\le 18 (m + m^{\prime}) - 26$ & $\le 22 (m + m^{\prime})$ \\ \hline
     \rule{0pt}{3ex}
     \rule[-1.5ex]{0pt}{0pt}
$\ell = 2$ & $6$ & $4$     & $10$     & $14$     & $\le 18 (m + m^{\prime}) - 10$   & $\le 22 (m + m^{\prime}) + 20$   \\ \hline
     \rule{0pt}{3ex}
     \rule[-1.5ex]{0pt}{0pt}
$\ell = 3$ & $\le 24$ & $12$     & $26$     & $38$     & $\le 18 (m + m^{\prime}) + 6$   & $\le 22 (m + m^{\prime}) + 40$   \\ \hline
     \rule{0pt}{3ex}
     \rule[-1.5ex]{0pt}{0pt}
$\ell = 4$ & $\le 40$ &$36$     & $\le 58$     & $\le 84$     & $\le 18 (m + m^{\prime}) + 22$   & $\le 22 (m + m^{\prime}) + 60$   \\ \hline
     \rule{0pt}{3ex}
     \rule[-1.5ex]{0pt}{0pt}     
$\ell \ge 5$ & $\le 16 \ell - 24$ & $\le 16 \ell - 24$     & $\le 16 \ell - 6$     & $\le 20 \ell + 4$     & $\le 18 (m + m^{\prime}) + 16 \ell - 42$   & $\le 22 (m + m^{\prime}) + 20 \ell - 20$   \\ \hline
\end{tabular}
\caption{
\textbf{$\CNOT$ count of multi-controlled $X$, $R_{y}$, $\RBS$ and $\gRBS$ gates used throughout this work.} 
Details are provided in Sec. \ref{app:compilation_one_gate}. 
For each multi-controlled $\RBS$ and $\gRBS$ gate the $\CNOT$ count corresponds to the best between the two compilations in Figs. \ref{fig:rbscomplex} and \ref{fig:generalizedRBS}. 
For multi-controlled-$R_{y}$ gates, we use the $\CNOT$ counts provided in \cite{Barenco1995, Vale2023}. 
The numbers in the second column are the $\chi_{\ell}$ in Lemma \ref{lemma:cRYcompilation}, while the second and third columns correspond to $C_{\ell}$ used in Eq. \eq{totalCNOTcount} for the real and complex cases, respectively. 
}
\label{tab:cgRBScompilation}
\end{table*}
Here, we calculate the $\CNOT$ cost of the $\RBS$ gate, the $\gRBS$ gate, and their multi-controlled versions. 
The following lemma will be useful. 

\begin{lemma}[$\CNOT$ cost of a multi-controlled $\operatorname{U}(2)$]
\label{lemma:cRYcompilation}
Let $\chi_{\ell}(U)$ be the number of $\CNOT$ gates to implement a $\ell$-controlled single-qubit gate $U \in \operatorname{U}(2)$, where $\ell = 0$ corresponds to no controls. 
Then $\chi_{0}(U) = 0$, $\chi_{1}(U) = 2$, $\chi_{2}(U) = 4$, $\chi_{3}(U) = 12$, $\chi_{4}(U) = 36$, and, for $\ell \ge 5$, $\chi_{\ell}(U) \le 20 \ell - 18$.
Moreover, if $U = R_{y}(\theta)$ or $U = R_{z}(\theta)$,the bound for $\ell \ge 5$ can be improved to $\chi_{\ell}(U) \le 16 \ell - 24$.
\end{lemma}

\begin{proof}
    Exact compilations for $\ell \le 4$ follow from Ref. \cite{Barenco1995,Iten2016}. 
    The upper bounds for $\ell \ge 5$ are given by Ref. \cite{Vale2023}.
\end{proof}

Depending on the chosen decomposition for the $\RBS$ gate, the $\gRBS$ gate admits different decompositions into controlled-$R_{y}$ and $R_z$ gates (\emph{e.g.} see Fig. \ref{fig:rbscomplex}). 
Using Lemma \ref{lemma:cRYcompilation}, we can now calculate the $\CNOT$ cost of a multi-controlled $\gRBS$ gate as follows. 

\begin{lemma}[$\CNOT$ cost of a multi-controlled complex $\gRBS$ gate]
\label{lemma:cgRBScompilationcomplex}
Let $\chi_{\ell}$ be as defined in Lemma \ref{lemma:cRYcompilation}, $c_{\vctrlbf}R^{\binbf}_{\toutbf}(\theta, \, \phi)\equiv c_{\mathcolor{violet}{ctrl_{1}, \ldots, ctrl_{\ell}}}R^{\mathcolor{blue}{in_{1}, \ldots, in_{m}}}_{\mathcolor{teal}{out_{1}, \ldots, out_{m^{\prime}}}}$ be as in Sec. \ref{sec:preliminaries}, and $\mu = \ell + m + m^{\prime}$. 
For $\phi = 0$, the number of $\CNOTs$ to compile this $\mu$-qubit gate is
\begin{align*}
2(m + m^{\prime} - 1) +\min\!\left(2 \chi_{\mu - 2}(R_{y}), \,  \chi_{\mu - 1}(R_{y})\right).
\end{align*}
If $\phi \ne 0$, then the cost is
\begin{align*}
2(m + m^{\prime} - 1) + \min\!\left(4 \chi_{\mu - 2}(R_{y}), \, \chi_{\mu - 1}(U)\right). 
\end{align*}
\end{lemma}

\begin{proof}
($i$) Case $\phi = 0$: Since the $R_{z}(\phi)$ gates are absent, the $\ell$ controls can act directly on the $R_{y}$ gates (see Fig. \ref{fig:generalizedRBS}). 
As a result, $c_{\mathcolor{violet}{ctrl_{1}, \ldots, ctrl_{\ell}}}R^{\mathcolor{blue}{in_{1}, \ldots, in_{m}}}_{\mathcolor{teal}{out_{1}, \ldots, out_{m^{\prime}}}}$ is equivalent to $2 (m - 1) + 2(m^{\prime} - 1) + 2$ $\CNOTs$ and two ($\mu - 2$)-controlled $R_{y}$ gates for the compilation in Fig. \ref{fig:rbscomplex} (\textit{Top}), or $2(m - 1) + 2(m^{\prime} - 1) + 2$ $\CNOTs$ and a single ($\mu - 1$)-controlled $R_{y}$ gate for the compilation in Fig. \ref{fig:rbscomplex} (\textit{Bottom}). 
The claim then follows from compiling the controlled $R_{y}$ gates using Lemma \ref{lemma:cRYcompilation} and choosing the less costly compilation between the two. 
\emph{Top} is the best compilation for $\mu \le 4$, while \emph{Bottom} is the best otherwise.

\noindent ($ii$) Case $\phi \ne 0$: When using the $\RBS$ compilation in Fig. \ref{fig:rbscomplex} (\textit{Top}) to build a complex $\gRBS$, it is necessary to perform $2$ $(\mu - 2)$-controlled $R_{y}$ gates as well as $2$ $(\mu - 2)$-controlled $R_{z}$ gates.
Since these gates have the same CNOT cost \cite{Vale2023}, we express the total cost as $4$ times the cost of one $(\mu - 2)$-controlled $R_{y}$ gate.
On the other hand, the product $R_{z}(\phi) \, R_{y}(\theta)$ in the compilation shown in Fig. \ref{fig:rbscomplex} (\textit{Bottom}) corresponds to a single $\operatorname{SU}(2)$ of the form 
$U = e^{i\lambda W}$, where $\lambda(\theta,\phi)=\arccos(\cos\theta\cos\phi)$ and $W(\theta,\phi)=\frac{1}{\sin\lambda(\theta,\phi)}(-\sin\theta\sin\phi\,X+\sin\theta\cos\phi\,Y+\cos\theta\sin\phi\,Z)$. 
The proof then proceeds identically to Case ($i$). 
\emph{Top} is the best compilation for $\ell = 0$ and $m = m^{\prime} = 1$, while \emph{Bottom} is the best otherwise.
\end{proof}

\noindent{Explicit} $\CNOT$ counts for all the gates used throughout this work are summarized in Table \ref{tab:cgRBScompilation}. 

\subsection{$\CNOT$ cost of sparse and binary encoders}
\label{app:cnot_cost_sparse_binary}

Using the results of Sec. \ref{app:compilation_one_gate}, we numerically investigate the $\CNOT$ cost of implementing the sparse and binary encoders from Secs. \ref{sec:sparse} and \ref{sec:binary}, respectively.

Fig. \ref{fig:plot_appendix}(a) compares, as a function of the number of qubits $n$, the average $\CNOT$ cost of the sparse encoder presented in Alg. \ref{alg:sparse-hw-enc} with the $\operatorname{QRAM}$-like sparse encoder from Ref. \cite{Veras2022} as well as the ancilla-free state preparation from Ref. \cite{Shende2006}.
While the former is implemented in the \texttt{qclib} package \cite{Araujo2023}, the latter is currently implemented in the \texttt{Qiskit} package \cite{qiskit2024}.
We used the \texttt{Qibo} package \cite{qibo2021} to implement Alg. \ref{alg:sparse-hw-enc}.
The $\CNOT$ cost is averaged over the encoding of $100$ random sparse vectors $\bf{y}$ of dimension $2^n$ and sparsity $s \in \{n, \, n^{2}\}$. 
Each $\bf{y}$ was obtained by sampling a Haar-random vector $\vecx\in\mathbb{C}^s$ and embedding it into a $2^n$-dimensional sparse data vector by associating each entry $x_j$ with a random computational basis state $\ket{b_j}$ sampled from the uniform distribution. 
Ref. \cite{Shende2006} displayed a scaling of $\order{2^{n}}$ in $\CNOT$-gate count independent of the sparsity $s$ (solid and dashed blue lines). 
Even though Ref. \cite{Shende2006} provides an ancilla-free method for state preparation, it is shown to not be the most suitable for the preparation of sparse states.
The sparse encoder from Ref. \cite{Veras2022} presented a scaling of $\sim \order{n^{2.2}}$ for $s = n$ (solid greed line), and as $\sim \order{n^{3.12}}$ for $s = n^{2}$ (dashed green line).
This is a significant improvement when compared with Ref. \cite{Shende2006}. However, the algorithm from Ref. \cite{Veras2022} requires $\order{n}$ ancillary qubits to harness that improvement.
Meanwhile, the sparse encoder in Alg. \ref{alg:sparse-hw-enc} scales as $\sim \order{n^{3.29}}$ for $s = n$ (solid purple line), and as $\sim \order{n^{4.14}}$ for $s = n^{2}$ (dashed purple line).
These scalings are $\sim \order{n}$ inferior to the ones from Ref. \cite{Veras2022} for both sparsity levels tested.
We highlight, however, that Alg. \ref{alg:sparse-hw-enc} is ancilla-free.
This points to a trade-off between ancillary qubits and $\CNOT$-gate cost.
Algorithm \ref{alg:sparse-hw-enc} allows one to have a sparse amplitude encoder that is both $\operatorname{poly}(n)$ in $\CNOT$ cost and ancilla-free, at the price of having to implement linearly deeper circuits.

Fig. \ref{fig:plot_appendix}(b) shows the asymptotic behavior $\sim 8 \, n \, 2^{n}$ of the total gate count for the binary encoder (see Lemma \ref{lemma:binary} in the main text).

\begin{figure*}[!t]
    \includegraphics[width=\textwidth]{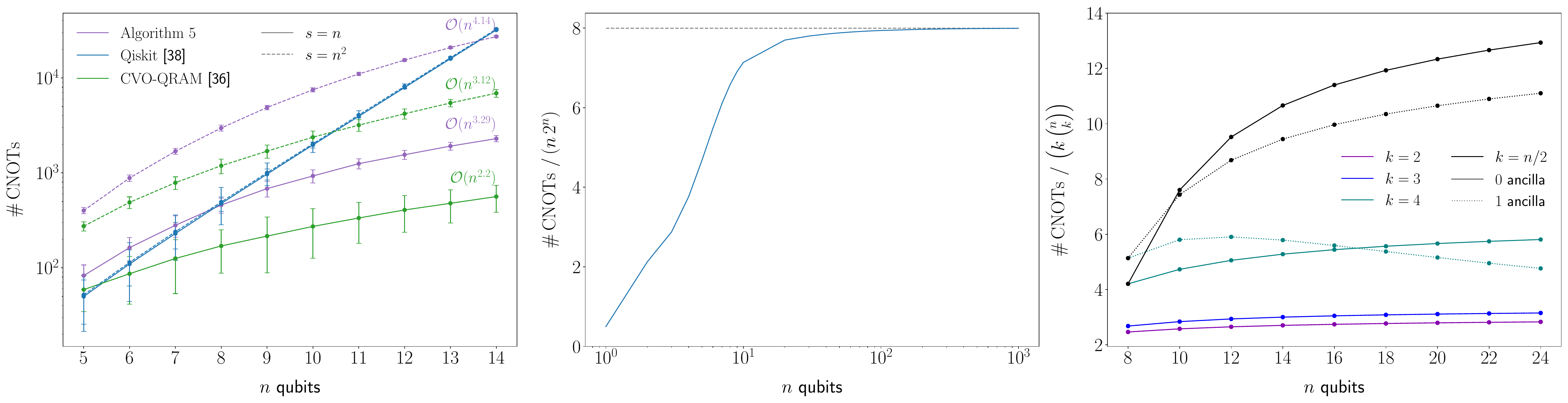}
    \caption{
      \emph{(left)} \textbf{Comparison of the average $\CNOT$ cost of Alg. \ref{alg:sparse-hw-enc} with Refs. \cite{Shende2006, Veras2022}.} 
      $\CNOT$-gate count (in log scale) as a function of the number of qubits $n$.
      We compare the costs of Alg. \ref{alg:sparse-hw-enc} (purple lines) with the algorithms from Refs. \cite{Veras2022} (green lines) and \cite{Shende2006} (blue lines).
      We present results averaged over $100$ Haar-random states $\in \mathbb{C}^{s}$, each embedded onto a data vector $\vecx \in \mathbb{C}^{2^{n}}$.
      Sparsity levels $s \in \{n, \, n^{2}\}$ are represented by solid and dashed lines, respectively.
      Error bars indicate one standard deviation.
      \emph{(middle)} \textbf{$\CNOT$-gate count of the binary encoder.} 
      Count based on Eq. \eqref{eq:cnotcountbinary} in Lemma \ref{lemma:binary} as a function of number of qubits $n$ (in log scale). 
      We used the compilation of controlled $R_{y}$ rotations from Tab. \ref{tab:cgRBScompilation}. 
      We observe an asymptotic behavior towards $8\,n\,2^n$ for large $n$.
      \emph{(right)} \textbf{$\CNOT$ cost of Alg. \ref{alg:hw-enc} with the addition of one ancilla.} 
      We compare the $\CNOT$ cost of Alg. \ref{alg:hw-enc} (solid lines) with the $\CNOT$ cost of the same algorithm with the additional step of including a clean ancillary qubit and using it to heuristically remove redundant controlled operations (dotted lines). 
      The comparison is performed for Hamming weights $k \in \{2, \, 3, \, 4, \, n/2 \}$
      For $k \in \{2, \, 3\}$ (purple and blue lines, respectively), the ancilla does not help decrease the $\CNOT$ significantly, if at all, independently of $n$. 
      For $k = 4$ (green lines), the addition of an ancillary qubit increases the number of $\CNOTs$ in the circuits for $n \leq 18$.
      This is due to the cost of extra controlled Toffoli gates required.
      For $n > 18$, the heuristic algorithm starts removing $\CNOTs$ of the original circuit, and the advantage with respect to the ancilla-free circuit increases with $n$.
      For $k = n/2$ (black lines), we can see that the contribution of the one clean ancilla appears at $n > 12$, and increases with $n$.
    }
    \label{fig:plot_appendix}
\end{figure*}

\section{Explicit examples}
\label{app:tables}

Here, we present three tables illustrating with explicit examples the execution of the classical algorithms used to generate the quantum circuits for our HW encoder algorithms. Each table corresponds to one of the figures present in the main text, namely: 
Table \ref{tab:example6choose2} illustrates the real dense encoder (Alg. \ref{alg:hw-enc}) and corresponds to Fig. \ref{fig:hwk2circuitexamples}; 
Table \ref{tab:example6choose3} illustrates the real sparse encoder (Alg. \ref{alg:sparse-hw-enc}) and corresponds to Fig. \ref{fig:hwk3circuitexamples}; 
finally, Table \ref{tab:hwktobinary} illustrates the binary encoder that combines the dense HW-$k$ for all possible $k\in[0,n]$ and corresponds to Fig. \ref{fig:hwtobincircuitexample}.

\section{Using one ancilla to reduce $\CNOT$ costs}
\label{app:cnc}

Here, we present a heuristic method that uses one ancilla to remove controlled operations in the later stages of the circuits generated by Algs. \ref{alg:hw-enc} and \ref{alg:hw-enc-cplx} as mentioned in Sec. \ref{sec:hamming_weight}.
The results are shown in Fig. \ref{fig:plot_appendix}(c). 
For a given $k$, the maximum number of indices in the $\vctrlbf$ set is $k-1$. 
Thus, consecutive controlled $\RBS$ gates that share the same $k - 1$ indices in their respective $\vctrlbf$ sets can be implemented together.
The \emph{\quotes{ancilla trick}} goes as follows.
First, we add one ancillary qubit to the system.
Then, a controlled Toffoli gate is added before the target consecutive controlled $\RBS$s.
This Toffoli gate acts on the ancilla and is controlled by the same qubits that are in the aforementioned $\vctrlbf$ set.
The next step is the replacement of the set of controls $\vctrlbf$ of the target $\RBS$ gates by $\vctrl = \{n + 1\}$, where $n + 1$ is the index of the newly added ancillary qubit.
Finally, the same Toffoli gate is added to the circuit after the stacked $\RBS$ gates, undoing the previous Toffoli operation.
The same \quotes{trick} can be applied to stacked $\RBS$ gates that have the same $k - 2$ indices in $\vctrl$, and so on.
The process ends when adding the two controlled Toffoli gates is more costly than implementing the original gates.

\section{Clifford Data Regression}
\label{app:cdr}

Clifford Data Regression (\textrm{CDR}) is a data-driven protocol to mitigate errors on expectation values estimated on \nisq devices \cite{czarnik2021, Lowe2021}. 
The protocol to implement \textrm{CDR} takes as inputs a quantum circuit $\mathcal{C}$ and an expectation value $\mu^{(0)}$ and repeats the following primitive:
\textit{(i)} construct a near-Clifford circuit by randomly replacing most, if not all, non-Clifford gates in the original circuit with Clifford gates;
\textit{(ii)} run both classical simulation and noisy implementation of the near-Clifford circuit, obtaining a set $\mathcal{S}$ containing new expectation values $\mu_{\textup{noiseless}}$ and $\mu_{\textup{noisy}}$ coming from the noiseless and noisy circuit, respectively, 
\begin{equation}\label{eq:cdr_set}
\mathcal{S} \coloneqq \left\{\left(\mu_{\textup{noisy}}^{(j)}, \, \mu_{\textup{noiseless}}^{(j)}\right)\right\}_{j \in \left[\,\abs{\mathcal{S}}\,\right]} \, ,
\end{equation}
where $\abs{\mathcal{S}}$ is the cardinality of the set $\mathcal{S}$.
The next step is to fit a regression model $f$ on $\mathcal{S}$ such that 
\begin{equation}
\mu_{\textup{noiseless}}^{(j)} \approx f_{\mathcal{S}}\left(\mu_{\textup{noisy}}^{(j)}\right), \, \forall \, j \in \left[\,\abs{\mathcal{S}}\,\right] \, .
\end{equation}
With that at hand, after executing the original circuit on the noisy hardware and obtaining the noisy expectation value $\mu_{noisy}^{(0)}$, the last step is to use the fitted model to get an error-mitigated version of $\mu_{\textup{noisy}}^{(0)}$, \emph{i.e.}
\begin{equation}
    \mu_{\textup{mitigated}}^{(0)} \approx f_{\mathcal{S}}\left(\mu_{\textup{noisy}}^{(0)}\right) \, .
\end{equation}
The error-mitigated expectation value $\mu_{\textup{mitigated}}^{(0)}$ is then the final empirical estimator of $\mu^{(0)}$.

In our proof-of-principle demonstration described in Sec.~\ref{sec:hardware}, we implemented \textrm{CDR} following the protocol detailed above training a linear model for each expectation value separately.
In doing so, we first compiled the original circuit in a way that the only non-Clifford gates present were $R_{y}$ rotations with at most one control qubit.
We call replacement rate, $r \in (0, 1)$, the number of $R_{y}$ gates replaced in the original circuit to generate each near-Clifford circuit $j$.
We created our set of near-Clifford circuits using replacement rates $r \in \{0.79, \, 0.82, \, 0.89, \, 0.93, \, 1.00 \}$ and sampling $300$ circuits per replacement rate, totaling $\abs{\mathcal{S}} = 1500$ data points.
The random Cliffords to be inserted were sampled uniformly from the set of all possible single- and two-qubit Clifford gates.
We also uniformly sampled the positions in the circuit of the $R_{y}$ gates that were replaced. 
The set \eq{cdr_set} was generated by classically simulating noiseless circuits locally while simulating the same circuits using the \ionq's proprietary noise models on \ionq's \textrm{Quantum Cloud} \cite{IONQCloud}.3
Every probability estimated from the hardware experiment was treated as an expectation value over a rank-$1$ projector in the computational basis and the mitigation procedure followed as described above by fitting one regression model per probability.

\begin{table}[t!]
\setlength{\tabcolsep}{0pt}
\begin{NiceTabular}
   [columns-width=1.1cm,hlines]
   {>{\rule[-2mm]{0pt}{.55cm}}*{1}{ccccccc}}
   $q_6$ & $q_5$ & $q_4$ & $q_3$ & $q_2$ & $q_1$ & Gate\\
   
   $\mathcolor{violet}{\overline{1}}$ & $\mathcolor{blue}{\overline{1}}$ & $0$ & $0$& $0$& $\mathcolor{teal}{0}$ &$\bcancel{c_{\mathcolor{violet}{6}}}\text{R}^{\mathcolor{blue}{5}}_{\mathcolor{teal}{1}}(\theta_1)$\\
   
   $\mathcolor{violet}{\overline{1}}$ & $0$ & $\overline{0}$ & $\overline{0}$& $\mathcolor{teal}{\overline{0}}$& $\mathcolor{blue}{1}$ &$\bcancel{c_{\mathcolor{violet}{6}}}\text{R}^{\mathcolor{blue}{1}}_{\mathcolor{teal}{2}}(\theta_2)$\\
   
   $\mathcolor{violet}{\overline{1}}$ & $0$ & $\overline{0}$ & $\mathcolor{teal}{\overline{0}}$& $\mathcolor{blue}{1}$& $0$ &$\bcancel{c_{\mathcolor{violet}{6}}}\text{R}^{\mathcolor{blue}{2}}_{\mathcolor{teal}{3}}(\theta_3)$\\
   
   $\mathcolor{violet}{\overline{1}}$ & $0$ & $\mathcolor{teal}{\overline{0}}$ & $\mathcolor{blue}{1}$& $0$& $0$ &$\bcancel{c_{\mathcolor{violet}{6}}}\text{R}^{\mathcolor{blue}{3}}_{\mathcolor{teal}{4}}(\theta_4)$\\
   
   $\mathcolor{blue}{\overline{1}}$ & $\mathcolor{teal}{0}$ & $\mathcolor{violet}{1}$ & $0$& $0$& $0$ &$c_{\mathcolor{violet}{4}}\text{R}^{\mathcolor{blue}{6}}_{\mathcolor{teal}{5}}(\theta_5)$\\
   
   $0$ & $\mathcolor{violet}{\overline{1}}$ & $\mathcolor{blue}{\overline{1}}$ & $0$& $0$& $\mathcolor{teal}{0}$ &$c_{\mathcolor{violet}{5}}\text{R}^{\mathcolor{blue}{4}}_{\mathcolor{teal}{1}}(\theta_6)$\\
   
   $0$ & $\mathcolor{violet}{\overline{1}}$ & $0$ & $\overline{0}$& $\mathcolor{teal}{\overline{0}}$& $\mathcolor{blue}{1}$ &$c_{\mathcolor{violet}{5}}\text{R}^{\mathcolor{blue}{1}}_{\mathcolor{teal}{2}}(\theta_7)$\\
   
   $0$ & $\mathcolor{violet}{\overline{1}}$ & $0$ & $\mathcolor{teal}{\overline{0}}$& $\mathcolor{blue}{1}$& $0$ &$c_{\mathcolor{violet}{5}}\text{R}^{\mathcolor{blue}{2}}_{\mathcolor{teal}{3}}(\theta_8)$\\
   
   $0$ & $\mathcolor{blue}{\overline{1}}$ & $\mathcolor{teal}{0}$ & $\mathcolor{violet}{1}$& $0$& $0$ &$c_{\mathcolor{violet}{3}}\text{R}^{\mathcolor{blue}{5}}_{\mathcolor{teal}{4}}(\theta_9)$\\
   
   $0$ & $0$ & $\mathcolor{violet}{\overline{1}}$ & $\mathcolor{blue}{\overline{1}}$& $0$& $\mathcolor{teal}{0}$ &$c_{\mathcolor{violet}{4}}\text{R}^{\mathcolor{blue}{3}}_{\mathcolor{teal}{1}}(\theta_{10})$\\
   
   $0$ & $0$ & $\mathcolor{violet}{\overline{1}}$ & $0$& $\mathcolor{teal}{\overline{0}}$& $\mathcolor{blue}{1}$ &$c_{\mathcolor{violet}{4}}\text{R}^{\mathcolor{blue}{1}}_{\mathcolor{teal}{2}}(\theta_{11})$\\
   
   $0$ & $0$ & $\mathcolor{blue}{\overline{1}}$ & $\mathcolor{teal}{0}$& $\mathcolor{violet}{1}$& $0$ &$c_{\mathcolor{violet}{2}}\text{R}^{\mathcolor{blue}{4}}_{\mathcolor{teal}{3}}(\theta_{12})$\\
   
   $0$ & $0$ & $0$ & $\mathcolor{violet}{\overline{1}}$& $\mathcolor{blue}{\overline{1}}$& $\mathcolor{teal}{0}$ &$c_{\mathcolor{violet}{3}}\text{R}^{\mathcolor{blue}{2}}_{\mathcolor{teal}{1}}(\theta_{13})$\\
   
   $0$ & $0$ & $0$ & $\color{blue}{\overline{1}}$& $\mathcolor{teal}{0}$& $\mathcolor{violet}{1}$ &$c_{\mathcolor{violet}{1}}\text{R}^{\mathcolor{blue}{3}}_{\mathcolor{teal}{2}}(\theta_{14})$\\
   
   $0$ & $0$ & $0$ & $0$& $1$& $1$ &\\
\CodeAfter
   \begin{tikzpicture} [-, shorten < = 0.7mm, shorten > = 0.7mm]
   \end{tikzpicture}
      \begin{tikzpicture} [-, shorten < = -1.2mm, shorten > = -1.2mm]
   \end{tikzpicture}
    \begin{tikzpicture} [->, shorten < = 1mm, shorten > = 1mm]
   \end{tikzpicture}
\end{NiceTabular}
\caption{
\textbf{Illustration of the dense HW-$k$ encoder (Alg. \ref{alg:hw-enc}) for $n = 6$ and $k = 2$.} 
Data vector $\vecx \in \mathbb{R}^d$, with $d = \binom{6}{2} = 15$. 
The order of the bitstrings is given by Alg. \ref{alg:nxt-hw-bs}. 
Marked bits are represented with an overline. 
Qubit indices are enumerated from right to left, $c_{\vctrl} R^{\bin}_{\tout}$ denotes an $\RBS$ on input qubit $\bin$, output qubit $\tout$ controlled by qubit(s) $\vctrl$, and $\theta_{j}$ given by Eq. \eq{spherical_coords}. 
Unnecessary $\vctrl$ qubits removed by Alg. \ref{alg:rbs-gate-params} are represented by $\bcancel{c_{\vctrl}}$. 
A gate in a given row represents what needs to be added to the circuit to add to the superposition the computational basis state represented by the bitstring in the next row.
}
\label{tab:example6choose2}
\end{table}

\begin{table}[t!]
\setlength{\tabcolsep}{0pt}
\begin{NiceTabular}[columns-width=1.1cm,hlines]{>{\rule[-2mm]{0pt}{.55cm}}*{7}{c}}
   \CodeBefore
   \Body
   $q_6$ & $q_5$ & $q_4$ & $q_3$ & $q_2$ & $q_1$ & Gate\\
   
   $0$ & $0$ & $\color{teal}{0}$ & $\color{blue}{1}$ & $\color{violet}{1}$ & $\color{violet}{1}$ &  $c_{\bcancel{\mathcolor{violet}{1}},\bcancel{\mathcolor{violet}{2}}}\text{R}^{\mathcolor{blue}{3}}_{\mathcolor{teal}{4}}(\theta_1)$\\
   
   $0$ & $0$ & $\color{violet}{1}$ & $\color{teal}{0}$ & $\color{violet}{1}$ & $\color{blue}{1}$ &  $c_{\bcancel{\mathcolor{violet}{2}},\mathcolor{violet}{4}}\text{R}^{\mathcolor{blue}{1}}_{\mathcolor{teal}{3}}(\theta_2)$\\

   ${0}$ & $\color{teal}{0}$ & $\color{blue}{1}$ & $\color{blue}{1}$ & $\color{violet}{1}$ & $\color{teal}{0}$ &  $c_{\bcancel{\mathcolor{violet}{2}}}\text{R}^{\mathcolor{blue}{3},\mathcolor{blue}{4}}_{\mathcolor{teal}{1},\mathcolor{teal}{5}}(\theta_3)$\\

   $0$ & $\color{violet}{1}$ & $\color{teal}{0}$ & $0$ & $\color{violet}{1}$ & $\color{blue}{1}$ &  $c_{\bcancel{\mathcolor{violet}{2}},\mathcolor{violet}{5}}\text{R}^{\mathcolor{blue}{1}}_{\mathcolor{teal}{4}}(\theta_4)$\\

   $\color{teal}{0}$ & $\color{blue}{1}$ & $\color{blue}{1}$ & $\color{teal}{0}$ & $\color{blue}{1}$ & $\color{teal}{0}$ &  $\text{R}^{\mathcolor{blue}{2},\mathcolor{blue}{4},\mathcolor{blue}{5}}_{\mathcolor{teal}{1},\mathcolor{teal}{3},\mathcolor{teal}{6}}(\theta_5)$\\

   $\color{violet}{1}$ & $\color{teal}{0}$ & $\color{teal}{0}$ & $\color{blue}{1}$ & $\mathcolor{teal}{0}$ & $\color{blue}{1}$ &  $c_{\mathcolor{violet}{6}}\text{R}^{\mathcolor{blue}{1},\mathcolor{blue}{3}}_{\mathcolor{teal}{2},\mathcolor{teal}{4},\mathcolor{teal}{5}}(\theta_6)$\\
   
   $1$ & $1$ & $1$ & $0$& $1$& $0$ & \\

\end{NiceTabular}
\caption{
\textbf{Illustration of the sparse encoder (Alg. \ref{alg:sparse-hw-enc}) for $n = 6$ and $s = 7$.} 
Data vector $\bf{y} \in \mathbb{R}^{d}$, with $d = 2^{6}$ and $s = 7$ non-zero entries.
The addresses $b_{j}$ are sorted in ascending order of HW (a necessary condition). As a byproduct, this sorting also minimizes the $\CNOT$ cost of Alg. \ref{alg:rbs-gate-params}. 
The angles $\theta_{j}$ are given by Eq.\eq{spherical_coords}. 
A gate in a given row represents what needs to be added to the circuit to add to the superposition the computational basis state represented by the bitstring in the next row.
The first $6$ elements have Hamming weight $k = 3$, leading to the initial controlled $\gRBS$ gates being HW-preserving. 
Meanwhile, the last one has $k = 4$ and requires a $\gRBS$ with $m \ne m'$ to increase the HW by $1$.
}
\label{tab:example6choose3}
\end{table}

\begin{table}[t!]
\setlength{\tabcolsep}{0pt}
\begin{NiceTabular}[columns-width=0.9cm,hlines]{>{\rule[-2mm]{0pt}{.55cm}}*{8}{c}}
   \CodeBefore
   \Body
   $q_6$ & $q_5$ & $q_4$ & $q_3$ & $q_2$ & $q_1$ & Operator & $k$\\

   $0$ & $0$ & $0$ & $0$ & $0$ & $0$ & $R^6_y(\theta_1)$ & $0$\\

   $\overline{1}$ & $0$ & $0$ & $0$ & $0$ & $0$ & $\LOAD_{B_1}$ & $1$\\

   $\vdots$ & $\vdots$ & $\vdots$ & $\vdots$ & $\vdots$ & $\vdots$ & $\mathcolor{gray}{\{\theta_2,...,\theta_6\}}$ & $1$\\

   $0$ & $1$ & $0$ & $0$ & $0$ & $0$ & $c_5 R^6_y(\theta_7)$ & $1$\\

   $\overline{1}$ & $\overline{1}$ & $0$ & $0$ & $0$ & $0$ & $\LOAD_{B_2}$ & $2$\\

   $\vdots$ & $\vdots$ & $\vdots$ & $\vdots$ & $\vdots$ & $\vdots$ & $\mathcolor{gray}{\{\theta_8,...,\theta_{21}\}}$ & $2$\\

   $0$ & $0$ & $0$ & $0$ & $1$ & $1$ & $c_{1,2} R^3_y(\theta_{22})$ & $2$\\

   $\overline{0}$ & $\overline{0}$ & $\overline{0}$ & $1$ & $1$ & $1$ & $\LOAD_{B_3}$ & $3$\\

   $\vdots$ & $\vdots$ & $\vdots$ & $\vdots$ & $\vdots$ & $\vdots$ & $\mathcolor{gray}{\{\theta_{23},...,\theta_{41}\}}$ & $3$\\

   $1$ & $1$ & $1$ & $0$ & $0$ & $0$ & $c_{4,5,6} R^3_y(\theta_{42})$ & $3$\\

   $\overline{1}$ & $\overline{1}$ & $\overline{1}$ & $\overline{1}$ & $0$ & $0$ & $\LOAD_{B_4}$ & $4$\\

   $\vdots$ & $\vdots$ & $\vdots$ & $\vdots$ & $\vdots$ & $\vdots$ & $\mathcolor{gray}{\{\theta_{43},...,\theta_{56}\}}$ & $4$\\

   $0$ & $0$ & $1$ & $1$ & $1$ & $1$ & $c_{1,2,3,4} R^5_y(\theta_{57})$ & $4$\\

   $\overline{0}$ & $1$ & $1$ & $1$ & $1$ & $1$ & $\LOAD_{B_5}$ & $5$\\

   $\vdots$ & $\vdots$ & $\vdots$ & $\vdots$ & $\vdots$ & $\vdots$ & $\mathcolor{gray}{\{\theta_{58},...,\theta_{62}\}}$ & $5$\\

   $1$ & $1$ & $1$ & $1$ & $1$ & $0$ & $c_{2,3,4,5,6} R^1_y(\theta_{63})$ & $5$\\

   $1$ & $1$ & $1$ & $1$ & $1$ & $1$ &  & $6$\\

\end{NiceTabular}
\caption{
\textbf{Illustration of the binary encoder for $n = 6$ qubits.} 
A gate (or collection of gates) in a given row represents what needs to be added to the circuit to add to the superposition the computational basis state represented by the bitstring in the next row.
The initial bitstring has $k = 0$ and (controlled) $R_{y}$ gates are used to increase the HW by $1$. 
The first bitstring with HW-($k+1$) has a one-bit difference from the final bitstring of the Ehrlich algorithm with HW-$k$.
Operators $\LOAD_{B_{k}}$ fill up the subspace of HW-$k$, and the algorithm stops after reaching $k = n$. 
Angles $\theta_{j}$ are given by Eq.\eq{spherical_coords}, $\theta_{j}$ are used in $R_{y}$ gates and $\mathcolor{gray}{\theta_{j}}$ are used in $\LOAD_{B_{k}}$. 
Total number of parameters is $2^{n} - 1$.
}
\label{tab:hwktobinary}
\end{table}

\end{document}